\documentclass[11pt]{article}

\usepackage[utf8]{inputenc}
\usepackage{amsmath}
\usepackage{amsthm}
\usepackage{amssymb}
\usepackage{bm}
\usepackage[makeroom]{cancel}
\usepackage{color}
\usepackage{enumitem}
\usepackage{epsfig}
\usepackage{float}
\usepackage{fullpage}
\usepackage[margin=1in]{geometry}
\usepackage{hyperref}
\usepackage{cleveref}
\usepackage{indentfirst}
\usepackage[numbers]{natbib}
\usepackage{mathrsfs}
\usepackage{mathtools}
\usepackage{mleftright}
\usepackage{relsize}
\usepackage{thmtools}
\usepackage{thm-restate}
\usepackage{verbatim}
\usepackage{./shortcuts}
\usepackage{./shortcuts_alphabet}

\usepackage{tikz}
\usetikzlibrary{patterns}

\usepackage[labelfont=bf]{caption}

\usepackage{algorithm}
\usepackage[noend]{algpseudocode}

\usepackage{times}

\definecolor{corlinks}{RGB}{0,0,150}
\definecolor{cormenu}{RGB}{0,0,150}
\definecolor{corurl}{RGB}{0,0,150}

\hypersetup{
	colorlinks=true,
	urlcolor=corlinks,
	linkcolor=corlinks,
	menucolor=cormenu,
	citecolor=corlinks,
	pdfborder= 0 0 0
}

\definecolor{blue-violet}{rgb}{0.54, 0.17, 0.89}

\newtheorem{theorem}{Theorem}
\newtheorem{lemma}[theorem]{Lemma}
\newtheorem{definition}[theorem]{Definition}
\newtheorem{fact}[theorem]{Fact}
\newtheorem{claim}[theorem]{Claim}
\newtheorem{remark}[theorem]{Remark}
\newtheorem{corollary}[theorem]{Corollary}
\newtheorem{proposition}[theorem]{Proposition}

\newcommand\hlb[1]{\tikz[overlay, remember picture,baseline=-\the\dimexpr\fontdimen22\textfont2\relax]\node[rectangle,fill=orange!50,fill opacity = 0.8,fill,text opacity =1, pattern=crosshatch, pattern color=orange] {$#1$};} 

\newcommand\hlg[1]{\tikz[overlay, remember picture,baseline=-\the\dimexpr\fontdimen22\textfont2\relax]\node[rectangle,fill=blue!60,fill opacity = 0.8,fill,text opacity =1, pattern=grid, pattern color=blue] {$#1$};} 

\newcommand\hlc[1]{\tikz[overlay, remember picture,baseline=-\the\dimexpr\fontdimen22\textfont2\relax]\node[rectangle,fill=gray!60,fill opacity = 0.8,fill,text opacity =1, pattern=north west lines, pattern color=gray] {$#1$};}

\DeclareFontFamily{U}{mathb}{\hyphenchar\font45}
\DeclareFontShape{U}{mathb}{m}{n}{<5> <6> <7> <8> <9> <10> gen * mathb
<10.95> mathb10 <12> <14.4> <17.28> <20.74> <24.88> mathb12}{}
\DeclareSymbolFont{mathb}{U}{mathb}{m}{n}

\DeclareMathSymbol{\circ}{\mathbin}{mathb}{'367}
\makeatletter
\newcommand*{\Rcirclearrow}{\mathpalette\@Rcirclearrow{}}
\newcommand*{\@Rcirclearrow}[1]{%
    \mathbin{\ooalign{\hphantom{$#1\circ$}\cr\hss\raisebox{0.05ex}{%
                \scalebox{0.9}{%
                    \rotatebox[origin=c]{270}{%
                        $#1\circ$}}}\hss}}}
\makeatother

\renewcommand\hlb[1]{\tikz[overlay, remember picture,baseline=-\the\dimexpr\fontdimen22\textfont2\relax]\node[rectangle,fill=orange!50,fill opacity = 0.8,fill,text opacity =1, pattern=crosshatch, pattern color=orange] {$#1$};} 

\renewcommand\hlg[1]{\tikz[overlay, remember picture,baseline=-\the\dimexpr\fontdimen22\textfont2\relax]\node[rectangle,fill=blue!60,fill opacity = 0.8,fill,text opacity =1, pattern=grid, pattern color=blue] {$#1$};} 

\renewcommand\hlc[1]{\tikz[overlay, remember picture,baseline=-\the\dimexpr\fontdimen22\textfont2\relax]\node[rectangle,fill=gray!60,fill opacity = 0.8,fill,text opacity =1, pattern=north west lines, pattern color=gray] {$#1$};}

\newcommand{\eqdef}{\stackrel{\rm def}{=}}

\def\colorful{1}

\ifnum\colorful=1

\fi
\ifnum\colorful=0

\fi

\newcommand{\amb}{\Gamma}

\newcommand{\F}{\mathcal{F}}
\newcommand{\B}{\mathcal{B}}
\newcommand{\G}{\mathcal{G}}

\newcommand{\R}{\mathcal{R}}

\begin{document}

\title{Boolean Circuit Complexity and Two-Dimensional Cover Problems\vspace{0.5cm}}

\author{Bruno
    P. Cavalar\vspace{0.2cm}\thanks{\texttt{E-mail:~bruno.cavalar@cs.oxford.ac.uk}}\\{\small
    Department of Computer Science}\\{\small University of Oxford}\vspace{0.4cm}
\and
Igor C. Oliveira\vspace{0.2cm}\thanks{\texttt{E-mail:~igor.oliveira@warwick.ac.uk}}\\{\small Department of Computer Science}\\{\small University of Warwick}\vspace{0.4cm}
}

\maketitle

\vspace{-0.6cm}

\begin{abstract}
We reduce the problem of proving deterministic and nondeterministic Boolean circuit size lower bounds to the analysis of certain two-dimensional combinatorial cover problems. This is obtained by combining results of Razborov (1989), Karchmer (1993), and Wigderson (1993) in the context of the fusion method for circuit lower bounds with the graph complexity framework of Pudl\'{a}k, R\"{o}dl, and Savick\'{y} (1988). For convenience, we formalize these ideas in the more general setting of ``discrete complexity'', i.e., the natural set-theoretic formulation of circuit complexity, variants of communication complexity, graph complexity, and other measures. 

We show that random graphs have linear graph cover complexity, and that explicit super-logarithmic graph cover complexity lower bounds would have significant consequences in circuit complexity. We then use discrete complexity, the fusion method, and a result of Karchmer and Wigderson (1993) to introduce nondeterministic graph complexity. This allows us to establish a connection between graph complexity and nondeterministic circuit complexity. 

Finally, complementing these results, we describe an exact characterization of the power of the fusion method in discrete complexity. This is obtained via an adaptation of a result of Nakayama and Maruoka (1995) that connects the fusion method to the complexity of 
 ``cyclic'' Boolean circuits, which generalize the computation of a  circuit by allowing cycles in its specification.

\end{abstract}

\newpage

\setcounter{secnumdepth}{3}
\setcounter{tocdepth}{2}

\vspace{-0.5cm}

\tableofcontents

\section{Introduction}\label{s:intro}

\subsection{Overview}

Obtaining circuit size lower bounds for explicit Boolean functions is a central research problem in theoretical computer science. While restricted classes of circuits such as constant-depth circuits and monotone circuits   are reasonably well understood (see, e.g., \citep{DBLP:books/daglib/0028687}), understanding the power and limitations of general (unrestricted) Boolean circuits remains a major challenge. 

The strongest known lower bounds on the number of gates necessary to compute an explicit Boolean function $f \colon \{0,1\}^n \to \{0,1\}$ are of the form $C \cdot n$ for a constant $C \leq 5$. The largest known value of $C$ depends on the exact set of allowed operations (see \citep{DBLP:conf/stoc/Li022, DBLP:conf/focs/FindGHK16} and references therein). To the best of our knowledge, the existing lower bounds on gate complexity for unrestricted Boolean circuits with a single output bit have all been obtained via the gate elimination method and its extensions. Unfortunately, it is not expected that this technique can lead to much better bounds \citep{DBLP:conf/mfcs/GolovnevHKK16}, let alone super-linear circuit size lower bounds. 

This paper revisits a classical approach to lower bounds known as the fusion method \citep{DBLP:conf/stoc/Razborov89, DBLP:conf/coco/Karchmer93}. The latter reduces the analysis of the circuit complexity of a Boolean function to obtaining bounds on certain related combinatorial cover problems. The method can also be adapted to weaker circuit classes, where it has been successful in some contexts (see \citep{Wigderson93thefusion} for an overview of results).\footnote{The fusion method can be seen as an instantiation of the generalized approximation method. For a self-contained  exposition of the connection between the fusion method and the approximation method, we refer the reader to \citep{notes_approx}.} 

An advantage of the fusion method over the gate elimination method is that it provides a tight characterization (up to a constant or polynomial factor, depending on the formulation) of the circuit complexity of a function. In particular, if a strong enough circuit lower bound holds, then in principle it can be established via the fusion method.

\paragraph{Contributions.} We can informally summarize our  contributions  as follows: 
\begin{enumerate}
    \item We exhibit a new instantiation of the fusion method that reduces the problem of proving deterministic and nondeterministic Boolean circuit size lower bounds to the analysis of ``two-dimensional'' combinatorial cover problems. 
    \item To achieve this, we introduce a framework that combines the fusion method for lower bounds with the notion of graph complexity and its variants \citep{DBLP:journals/acta/PudlakRS88, juknacomputational}. In particular, we observe that cover complexity offers a particularly strong ``transference'' theorem between Boolean circuit complexity and graph complexity. 
    \item As a byproduct of our conceptual and technical contributions, we obtain a tight asymptotic bound on the cover complexity of a random graph, and introduce a useful notion of nondeterministic graph complexity.
    \item Finally, we describe an exact correspondence between cover complexity and circuit complexity. This is  relevant for the investigation of state-of-the-art circuit lower bounds of the form $C \cdot n$, where $C$ is constant.
\end{enumerate}

\noindent In the next section, we describe these results and their connections to previous work in more detail.

\subsection{Results}

\paragraph{Notation.} Given a family $\mathcal{B} = \{B_1, \ldots, B_m\}$, where each set $B_i$ is contained in a finite fixed ground set $\Gamma$, and a target set $A$, we let $D(A \mid \mathcal{B})$ denote the minimum total number of pairwise unions and intersections needed to construct $A$ starting from $B_1, \ldots, B_m$. We say that $D(A \mid \mathcal{B})$ is the \emph{discrete complexity} of $A$ with respect to $\mathcal{B}$ (see \Cref{ss:disc_defin_notat} for a formal presentation). We will be interested in the discrete complexity of \emph{non-trivial} sets $A$, i.e., when $A \neq \emptyset$ and $A \neq \Gamma$. 

This general definition can be used to capture a variety of  problems. For
instance, the monotone circuit complexity of a function $f \colon \{0,1\}^n
\to \{0,1\}$ is simply $D(f^{-1}(1) \mid \{x_1, \ldots, x_n, \emptyset,
\bar{1}\})$, where each symbol from $\{x_1, \ldots, x_n, \emptyset,
\bar{1}\}$ represents the natural corresponding  subset of $\{0,1\}^n$.
Similarly, we can capture (non-monotone) Boolean circuit complexity by
considering the family $\mathcal{B}_n = \{x_1, \ldots, x_n, \overline{x_1},
\ldots, \overline{x_n}\}$ of subsets of $\{0,1\}^n$ and the corresponding
complexity measure $D(f^{-1}(1) \mid \B_n)$.\footnote{This captures the
\emph{DeMorgan circuit complexity}, where negations are at the bottom of
the circuit.}

Let $N = 2^n$ for some $n \in \mathbb{N}$, and let $[N] = \{1,2,\ldots, N\}$. As another example in discrete complexity, we can consider subsets $R_1, \ldots, R_N, C_1, \ldots, C_N$ of the ground set $[N] \times [N]$, where each set $R_i = \{ (i,j) \mid j \in [N]\}$ corresponds to the $i$-th ``row'', and each set $C_j = \{ (i,j) \mid i \in [N]\}$ corresponds to the $j$-th ``column''. Then, given a set $G \subseteq [N] \times [N]$ and $\G_{N,N} = \{R_1, \ldots, R_N, C_1, \ldots, C_N\}$, the quantity $D(G \mid \G_{N,N})$ is known as the \emph{graph complexity} of $G$ (see \citep{DBLP:journals/acta/PudlakRS88, juknacomputational}). 

For the discussion below, we will need another definition. We let $D_\cap(A
\mid \mathcal{B})$ denote the minimum number of pairwise intersections
sufficient to construct $A$ from the sets in $\B$. We say that $D_\cap(A
\mid \mathcal{B})$ is the \emph{intersection complexity} of $A$ with
respect to $\B$. 
When $\B = \B_n$, we may refer to intersection complexity with respect to
$\B$ as \emph{AND complexity}.
We refer to \Cref{fig:chessboard} for an example. It is
possible to show that $D_\cap(A \mid \B)$ and $D(A \mid \B)$ are
polynomially related, with a dependency on $|\B|$ (see
\Cref{ss:basic_results} for more details).

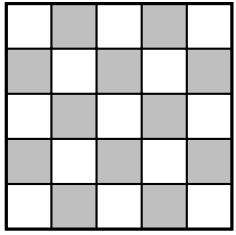
\begin{figure}[htbp]
  \centering
\begin{tikzpicture}[scale=0.6]
  \foreach \x in {0, 1, 2, 3, 4} {
    \foreach \y in {0, 1, 2, 3, 4} {
      \fill[lightgray] (\x, \y) rectangle ++(1, 1);
      \ifodd\x\relax
        \ifodd\y\relax
          \fill[white] (\x, \y) rectangle ++(1, 1);
        \fi
      \else
        \ifodd\y\relax\else
          \fill[white] (\x, \y) rectangle ++(1, 1);
        \fi
      \fi
    }
  }

    \draw[black, thick] (0,0) grid (5,5);
    
    \draw[very thick] (0, 0) rectangle (5, 5);
\end{tikzpicture}
\caption{A graphical representation of a set $G \subseteq [5] \times [5]$ of intersection complexity $D_\cap(G \mid \G_{5,5}) \leq 2$ via $G = \big ((R_2 \cup R_4) \cap (C_1 \cup C_3 \cup C_5) \big ) \cup \big ((C_2 \cup C_4) \cap (R_1 \cup R_3 \cup R_5)\big )$.}
  \label{fig:chessboard}
\end{figure}

Given an arbitrary set $A$ and a family $\B$ as above, one can introduce a
complexity measure $\rho(A, \B)$ that is closely related to $D(A \mid \B)$.
In more detail, we define an appropriate bipartite graph $\Phi_{A, \B} =
(V_\mathsf{pairs}, V_\mathsf{filters}, \mathcal{E})$, called the
\emph{cover graph} of $A$ and $\B$, and let $\rho(A, \B)$ denote the
minimum number of vertices in $V_\mathsf{pairs}$ whose adjacent edges cover
all the vertices in $V_\mathsf{filters}$. (Since the definition of the
graph $\Phi_{A, \B}$ is somewhat technical and won't be needed in the
subsequent discussion, it is deferred to \Cref{ss:fusion_defin_notat}). We
say that $\rho(A,\B)$ is the \emph{cover complexity} of $A$ with respect to
$\B$.
This measure of complexity generalises the cover problem introduced by~\cite{DBLP:conf/coco/Karchmer93, Wigderson93thefusion} to capture circuit complexity.
\\

Our first observation is that, by a straightforward adaptation of the fusion method for lower bounds  \citep{DBLP:conf/stoc/Razborov89, DBLP:conf/coco/Karchmer93, Wigderson93thefusion} to our framework, the following relation holds:
\begin{equation} \label{eq:intro_1}
    \rho(A, \B) \;\leq\; D_\cap(A \mid \B) \;\leq\; \rho(A,\B)^2.
\end{equation}
In particular, cover complexity provides a lower bound on intersection
complexity. We are particularly interested in applications of the
inequalities above to graph complexity. There are two main reasons for
this. Firstly, to each graph $G \subseteq [N] \times [N]$ one can associate
a natural Boolean function $f_G \colon \{0,1\}^n \times \{0,1\}^n \to
\{0,1\}$ (see \Cref{sec:transference}), where $N = 2^n$, and it is known
that lower bounds on the graph complexity of $G$ yield lower bounds on the
Boolean circuit complexity of $f_G$ \citep{DBLP:journals/acta/PudlakRS88}.
(There can be a significant loss on the parameters of such transference
results depending on the context. We refer to \citep{juknacomputational}
for more details. See also the discussion before \Cref{rk:transf} below.) 
Secondly, the cover problem defining
$\rho(G, \G_{N,N})$ involves a two-dimensional ground set $[N] \times [N]$,
in contrast to the $n$-dimensional ground set $\{0,1\}^n$ found in Boolean
function complexity. We hope this perspective can inspire new techniques,
and indeed we show how this perspective can be used to give a tight
bound for a natural Boolean function in \Cref{ss:example_fusion_lb}.

Our second observation is that a tight connection can be established between graph complexity and Boolean circuit complexity by focusing on intersection complexity and cover complexity.

\begin{lemma}[Transference of Lower Bounds]\label{lem:transf_intro}
    For every non-trivial bipartite graph $G \subseteq [N] \times [N]$ and  corresponding Boolean function $f_G \colon \{0,1\}^{n} \times \{0,1\}^n \to \{0,1\}$, we have
    \begin{eqnarray}
        \rho(f_G^{-1}(1), \B_{2n}) & \geq & \rho(G, \G_{N,N}),~\text{and}   \\
         D(f_G^{-1}(1) \mid \B_{2n}) & \geq & D_\cap(G \mid \G_{N,N}).  
    \end{eqnarray}
\end{lemma}

The second inequality is implicit in the literature on graph complexity. We include it in the statement of \Cref{lem:transf_intro} for completeness. Using \Cref{lem:transf_intro}, \Cref{eq:intro_1}, and another idea, we note in \Cref{sec:transference} that a lower bound of the form $C \cdot \log N$ on $\rho(G, \G_{N,N})$ yields a lower bound of the form $C \cdot m - O(1)$ on the 
AND
complexity of a related function $F \colon \{0,1\}^m \to \{0,1\}$. 
It is worth noting that lower bounds of  the form $Cn$ for $C > 1$ on the
AND complexity of explicit Boolean functions can be obtained using
gate-elimination techniques \citep{GolovnevPrivate18}, so the
problem
considered here 
does not suffer from a  ``barrier'' at $n$ gates as in the setting of
multiplicative complexity \citep{DBLP:conf/aaecc/Schnorr88}. We leave open
the problem of matching (or more ambitiously strengthening) existing
Boolean circuit lower bounds obtained via gate elimination using our
framework.

Complementing the approach to non-trivial circuit lower bounds discussed above, we show the following result for non-explicit graphs.

\begin{theorem}[Cover complexity of a random graph]\label{thm:intro_random_graph} Let $N = 2^n$, and let 
   $G \sseq [N] \times [N]$ be a uniformly random bipartite graph.
    Then, asymptotically almost surely,
$$
\rho(G, \G_{N,N}) \; = \; \Theta(N).
$$ 
\end{theorem}

Since the state of the art in Boolean circuit lower bounds is of the form $C \cdot n$ for a small constant $C$, the discussion above motivates the investigation of a tighter version of \Cref{eq:intro_1}. Next, we show that cover complexity can be \emph{exactly} characterized using the  complexity of \emph{cyclic constructions}. Roughly speaking,  $D^\circ(A \mid \mathcal{B})$ denotes the minimum number of unions and intersections in a cyclic construction of $A$ from sets in $\B$, where a cyclic construction can be seen as the analogue of a Boolean circuit allowed to contain cycles. We refer to \Cref{ss:cyclic} for the definition. Similarly, we can also consider  $D_\cap^\circ(A \mid \mathcal{B})$, the intersection complexity of cyclic constructions.

\begin{theorem}[Exact characterization of cover complexity]\label{thm:intro_fusion_cyclic}
Let $A \subseteq \amb$ be a non-trivial set, and let $\B \subseteq \mathcal{P}(\amb)$ be a non-empty family of sets. Then
$$
\rho(A, \B) \;= \; D_\cap^\circ(A \mid \mathcal{B}).
$$
\end{theorem}

This precise correspondence is obtained by refining an idea from \citep{DBLP:journals/iandc/NakayamaM95}, which  obtained a characterization of a variant of cover complexity up to a constant factor. There are some technical differences though. In contrast to their work, here we consider (monotone) semi-filters instead of a more general class of functionals $\mathcal{F} \subseteq \mathcal{P}(U)$ in the definition of cover complexity, and intersection complexity instead of Boolean circuit complexity. Additionally, the result is presented in the set-theoretic framework of the fusion method (which is closer to our notion of discrete complexity), while \citep{DBLP:journals/iandc/NakayamaM95} employed a formulation via legitimate models and the generalized approximation method.

As an immediate consequence of \Cref{thm:intro_fusion_cyclic} and a cover complexity lower bound from \citep{DBLP:conf/coco/Karchmer93}, it follows that every monotone \emph{cyclic} Boolean circuit that decides if an input graph on $n$ vertices contains a triangle contains at least $\Omega(n^3/(\log n)^4)$ fan-in two $\mathsf{AND}$ gates.\footnote{This consequence does not immediately follow from the work of \citep{DBLP:journals/iandc/NakayamaM95}, as their formulation is not consistent with the use of monotone functionals employed in the definition of $\rho$ followed here and in \citep{DBLP:conf/coco/Karchmer93}.} We refer to \Cref{ss:fusion_cyclic} for more details.

The tight bound in \Cref{thm:intro_fusion_cyclic} highlights a mathematical advantage of the investigation of cyclic constructions and cyclic Boolean circuits. Interestingly, the strongest known lower bounds against unrestricted (non-monotone) Boolean circuits obtained via the gate elimination method \citep{DBLP:conf/stoc/Li022, DBLP:conf/focs/FindGHK16} also incorporate concepts related to cyclic computations.

Our last contribution is of a conceptual nature. The fusion method offers a different yet equivalent  formulation of circuit complexity. This allows us to port some of the  abstractions and characterizations provided by different notions of cover complexity to the setting of discrete complexity. As an example, we introduce \emph{nondeterministic graph complexity} through a dual notion of ``nondeterministic'' cover complexity from \citep{DBLP:conf/coco/Karchmer93}, and show a simple application to nondeterministic Boolean circuit lower bounds via a transference lemma for nondeterministic complexity.\footnote{Observe that the definition of nondeterministic complexity for Boolean functions relies on Boolean circuits extended with extra input variables. It is not  obvious how to introduce a natural analogue in the context of graph complexity, which relies on graph constructions.}

Going beyond the contrast between state-of-the-art lower bounds for monotone and non-monotone computations, it would also be interesting to obtain an improved understanding of which settings of discrete complexity are susceptible to strong unconditional lower bounds.

\paragraph{Organization.} The main definitions  are given in \Cref{s:disc_complexity}. To make the paper self-contained, we include a proof of \Cref{eq:intro_1} in \Cref{s:fusion}. The proof of \Cref{lem:transf_intro} appears in \Cref{sec:transference} and \Cref{ss:fram_cover_complex}. The proof of \Cref{thm:intro_random_graph} is presented in \Cref{ss:fram_cover_complex}, while the proof of  \Cref{thm:intro_fusion_cyclic} is given in \Cref{ss:fusion_cyclic}. Finally, a discussion on nondeterministic graph complexity and a simple application of this notion appear  in \Cref{ss:nondet_graph_complexity}.

\paragraph{Acknowledgements.} We would like to thank Sasha Golovnev and Rahul Santhanam for discussions about the AND complexity of Boolean functions. This work received support from the Royal Society University Research Fellowship URF$\setminus$R1$\setminus$191059
and
URF$\setminus$R1$\setminus$211106; the UKRI Frontier Research Guarantee Grant EP/Y007999/1; and the Centre for Discrete Mathematics and its Applications (DIMAP) at the University of Warwick.

\section{Discrete Complexity}\label{s:disc_complexity}

\subsection{Definitions and notation}\label{ss:disc_defin_notat}

We adopt the convention that $\mathbb{N} \eqdef \{0,1,2, \ldots\}$, $\mathbb{N}^+ \eqdef \mathbb{N} \setminus \{0\}$, $[t] \eqdef \{1, \ldots, t\}$, where $t \in \mathbb{N}^+$, and $\mathcal{P}(\cdot)$ is the power-set construction. 

Let $\amb$ be a nonempty finite set. We refer to this set as the \emph{ground set}, or the \emph{ambient space}. Let $\mathcal{B} = \{B_1, \ldots, B_m\}$ be a family of subsets of $\amb$. We say that a set $B_i \in \B$ is a \emph{generator}. Given a set $A \subseteq \amb$, we are interested in the minimum number of elementary set operations necessary to construct $A$ from the generator sets in $\mathcal{B}$. The allowed operations are \emph{union} and \emph{intersection}. 
Formally, we let $D(A \mid \mathcal{B})$ be the minimum number $t \geq 1$ such that there exists a \emph{sequence} $A_1, \ldots, A_t$ of sets contained in $\amb$ for which the following holds: $A_t = A$, and for every $i \in [t]$, $A_i$ is either the union or the intersection of two (not necessarily distinct) sets in $\mathcal{B} \cup \{A_1, \ldots, A_{i-1}\}$. We say that a sequence of this form \emph{generates} $A$ from $\B$. If there is no finite $t$ for which such a sequence exists, then $D(A \mid \mathcal{B}) \eqdef \infty$.\footnote{A simple example is that of a non-monotone Boolean function represented by $A \subseteq \{0,1\}^n$ and $\B$ as the family of generators in monotone circuit complexity.} Consequently, $D \colon \mathcal{P}(\amb) \times \mathcal{P}(\mathcal{P}(\amb)) \to \mathbb{N}^+ \cup \{\infty\}$. We say that $D(A \mid \B)$ is the \emph{discrete complexity} of $A$ with respect to $\B$.

We use $D_\cap(A \mid \mathcal{B})$ to denote the minimum number of \emph{intersections} in any sequence that generates $A$ from $\B$. The value $D_\cup(A \mid \B)$ is defined analogously. We will often refer to these measures as 
\emph{intersection complexity} and \emph{union complexity}, respectively.
Given sets $A_1,\dots,A_s \sseq \amb$,
we will write $D(A_1,\dots,A_s \mid \calb)$ to refer
to the minimum length of a sequence that generates all the sets $A_i$;
the measure $D_{\cap}(A_1,\dots,A_s \mid \calb)$ is defined analogously.

\begin{fact}\label{f:first_trivial_fact}
If $A \in \mathcal{B}$, then $D(A \mid \mathcal{B}) = 1$ and $D_\cap(A \mid \B) = D_\cup(A \mid \B) = 0$.
\end{fact}

We have the following obvious inequality, which in general does not need to be tight (Fact \ref{f:first_trivial_fact} offers a trivial example).

\begin{fact}\label{f:additivity}
$D(A \mid \B) \geq D_\cap(A \mid \B) + D_\cup(A \mid \B)$.
\end{fact}

When the ambient space $\amb$ is clear from the context, we let $E^c \subseteq \amb$ denote the complement of a set $E \subseteq \amb$. For convenience, for a set $U \subseteq \amb$, we use $B_U$ as a shorthand for $B \cap U$. For a family of sets $\B$, we let $\B_U \eqdef \{B_U \mid B \in \B\}$. 

Let $A_1, \ldots, A_t$ be a sequence of sets that generates $A$ from $\B$, where $|\B| = m$. It will be convenient in some inductive proofs to consider the \emph{extended sequence} $B_1, \ldots, B_m, A_1, \ldots, A_t$ that includes as a prefix the generators from $\B$. The particular order of the sets $B_i$ is not relevant. While the extended sequence has length $m + t$, we will refer to it as a sequence of complexity $t$. Similarly, if the number of intersections employed in the definition of the sequence is $k$, we say it has intersection complexity $k$.

Given a construction of $A$ from $\B$ specified by a sequence $A_1, \ldots, A_t$ and its corresponding union and intersection operations, we let $\Lambda$ be the \emph{set of intersections} in the sequence, where we represent an intersection operation $A_\ell = A_i \cap A_j$ by the pair $(A_i, A_j)$.

For an ambient space $\amb$ and $\B \subseteq \mathcal{P}(\amb)$, we use $\langle \amb, \B \rangle$ to represent the corresponding \emph{discrete space}. We assume for simplicity that $\amb = \bigcup_{B \in \B} B$. We extend the notation introduced above, and use $D(A_1, \ldots, A_\ell \mid \B)$ to denote the discrete complexity of simultaneously generating $A_1, \ldots, A_\ell$ from $\B$. In other words, this is the minimum number $t$ such that there exists a sequence $E_1, \ldots, E_t$ of sets contained in $\amb$ such that every set $A_i$ appears in the sequence at least once, and each $E_j$ is obtained from the preceding sets in the sequence and the sets in $\B$ either by a union or by an intersection operation.

Finally, note that we tacitly assume in most  proofs presented in this section that $D(A \mid \B)$ is finite, as otherwise the corresponding statements are trivially true. We will also assume in these statements that $A \subseteq \bigcup_{B \in \B} B = \amb$ in order to avoid trivial considerations.

\subsection{Examples}

\subsubsection{Boolean circuit complexity}\label{ss:circuit_complexity}

This is the classical setting where for each $n \in \mathbb{N}^+$, $\amb = \{0,1\}^n$ is the set of vertices of the $n$-dimensional hypercube, $A$ corresponds to $f^{-1}(1)$ for a Boolean function $f \colon \{0,1\}^n \to \{0,1\}$, and $\B = \{B_1, \ldots, B_n, B_1^c, \ldots, B_n^c\}$, where $B_i = \{ v \in \amb \mid v_i = 1\}$. By definition, $D(A \mid \B)$ captures the \emph{circuit complexity} of $f$.  If we drop the generators $B_i^c$ from the family $\B$, and add the sets $\emptyset$ and $\bar{1} \eqdef \{0,1\}^n$ to it, we get \emph{monotone circuit complexity} instead of circuit complexity.

\subsubsection{Bipartite graph complexity}\label{ss:graph_complexity}

Let $\amb = [N] \times [M]$, where $N,M \in \mathbb{N}^+$. A set $G \subseteq \amb$ can be viewed either as a bipartite graph with parts $L = [N]$ and $R = [M]$, or as an $N \times M$ $\{0,1\}$-valued matrix. We let $R_i \subseteq [N] \times [M]$ denote the matrix with $1$'s in the $i$-th row, and $0$'s elsewhere. Similarly, $C_j \subseteq [N] \times [M]$ denotes the matrix with $1$'s in the $j$-th column, and $0$'s elsewhere. (Each $R_i$ and $C_j$ is called a \emph{star} in graph terminology). We let $\mathcal{G}_{N,M} = \{R_1, \ldots, R_N, C_1, \ldots, C_M\}$. The value $D(G \mid \mathcal{G}_{N,M})$ is known as the \emph{star complexity} of $G$ (\citep{DBLP:journals/acta/PudlakRS88}, see also \citep{juknacomputational} and references therein). We will refer to it simply as \emph{graph complexity}. Notice that, for every non-empty graph $G$, $D_\cap(G \mid \G_{N,M}) \leq \min \{N,M\}$.

We remark that a related notion of \emph{clique complexity} is discussed
in~\cite{DBLP:books/daglib/0028687}. In this notion,
the generators are
sets of the form
$W_S := \bigcup_{i \in S} R_i$
and
$Z_T := \bigcup_{j \in T} C_j$,
for some $S \sseq [N]$ and $T \sseq [M]$.
Let $\calk_{N,M} = \set{W_S : S \sseq [N]} \cup \set{Z_T : T \sseq
[M]}$.
Note that
the intersection clique complexity
of a graph $G$
is \emph{equal} to its intersection graph complexity
(i.e.,
$D_\cap(G \mid \calk_{N,M})
= 
D_\cap(G \mid \calg_{N,M})$).\footnote{
    We also
    remark that 
    the 
    \emph{decision tree clique complexity}
    of a graph $G$
    (in which we are allowed to query an arbitrary generator from
    $\calk_{N,M}$)
    is known to capture \emph{exactly} the 
    communication complexity of an associated function
    $f_G$~\cite[Section 3]{DBLP:journals/acta/PudlakRS88}.
}

One can also consider the graph complexity of \emph{non-bipartite} graphs
via an appropriate choice of generators~(as in, e.g.,~\cite{juknacomputational}),
though we will not be
concerned
with this variant in this work.

\subsubsection{Higher-dimensional generalizations of graph complexity}\label{ss:d_dim_graph_complexity}

This is the natural extension of the ambient space $[N] \times [N]$ to $[N]^d$, where $d \in \mathbb{N}^+$ is a fixed dimension. Every generator contained in $[N]^d$ is a set of elements described by a sequence of the form $(\star, \ldots, \star, a, \star , \ldots, \star)$, where an element $a \in [N]$ is fixed in exactly one coordinate. We let $\G_N^{(d)}$ be the corresponding family of generators. Notice that $|\G^{(d)}_N| = dN$. Given a $d$-dimensional tensor $A \subseteq [N]^d$, we denote its $d$\emph{-dimensional graph complexity} by $D(A \mid \G_N^{(d)})$.

To some extent, graph complexity and Boolean circuit complexity are extremal examples of non-trivial discrete spaces, in the sense that the former minimizes the number of dimensions and maximizes the possible values in each coordinate, while the latter does the opposite. The higher dimensional graphs generalize both cases.

\subsubsection{Combinatorial rectangles from communication complexity}\label{ss:comb_rect}

The domain is $[N] \times [N]$, and its associated family $\R_{N,N}$ of generators contains every \emph{combinatorial rectangle} $R = U \times V$, where $U, V \subseteq [N]$ are arbitrary subsets. In particular, $|\R_{N,N}| = 2^{2N}$, while the number of subsets of $[N] \times [N]$ is $2^{N^2}$. Observe that $\R_{N,N}$ extends the set of generators employed in graph complexity. Consequently, for $G \subseteq [N] \times [N]$, $D(G \mid \R_{N,N}) \leq D(G \mid \G_{N,N})$. Moreover, $D_\cap(G \mid \R_{N,N}) = 0$ for every graph.\\

Observe that there is an interesting contrast among all these different spaces: the ratio between the \emph{size of the ambient space} and \emph{the number of generators}. For instance, in graph complexity the two are polynomially related, in Boolean circuits the ambient space is exponentially larger, and in the discrete space involving combinatorial rectangles the opposite happens. These natural discrete spaces exhibit three important regimes of parameters in discrete complexity.

\subsection{Basic lemmas and other useful results}\label{ss:basic_results}

By combining sequences, we have the following trivial inequality. 

\begin{fact}\label{f:chain_rule}
For every set $E \subseteq \amb$ and $\diamond \in \{\cap, \cup\}$, $D_\diamond(A \mid \mathcal{B}) \leq D_\diamond(A \mid E, \mathcal{B}) + D_\diamond(E \mid \mathcal{B})$.\footnote{We often abuse notation and write $D(A \mid E, \B)$ instead of $D(A \mid \{E\} \cup \B)$.}
\end{fact}

\begin{proof}
Let $t_1 = D_\diamond(A \mid E, \mathcal{B})$, witnessed by the sequence $A_1, \ldots, A_{t_1}$. Also, let $t_2 = D_\diamond(E \mid \mathcal{B})$, with a corresponding sequence $E_1, \ldots, E_{t_2}$. Then $E_1, \ldots, E_{t_2}, A_1, \ldots, A_{t_1}$ is a sequence of length $t_1 + t_2$ showing that $D_\diamond(A \mid \mathcal{B}) \leq t_1 + t_2$. 
\end{proof}

Observe that a construction of an arbitrary set $A$ from $\B$ provides a construction of $A_U$ from the sets in $\B_U$ (recall that $A_U \eqdef A \cap U$, etc.). Indeed, it is easy to see that if $A^1, \ldots, A^t$ generates $A$ from $\B$, then $A^1_U, \ldots, A^t_U$ generates $A_U$ from $\B_U$.

\begin{fact}
$D(A_U \mid \mathcal{B}_U) \leq D(A \mid \B).$
\end{fact}

For convenience, we say that $A^1_U, \ldots, A^t_U$ is the \emph{relativization} of the sequence $A^1, \ldots, A^t$ with respect to $U$.

The following simple technical fact will be useful. The proof is an easy induction via extended sequences.

\begin{fact}\label{f:use_empty}
If $A$ is non-empty, 
then $D_\cap(A \mid \B) = D_\cap(A \mid \B \cup \{\emptyset\})$.
\end{fact}

The next lemma shows that intersection complexity and discrete complexity are polynomally related, with a dependency on $|\B|$. This was first observed for monotone circuits in \citep{DBLP:journals/combinatorica/AlonB87}.

\begin{lemma}[Immediate from \citep{DBLP:journals/ipl/Zwick96}]\label{l:zwick}
If $1 < D_\cap(A \mid \B) = k < \infty$, then $$D(A \mid \B) = O(k(|\B| + k)/ \log k).$$
\end{lemma}

We describe a self-contained, indirect proof of a weaker form of this lemma in Section \ref{ss:fusion_upper} (Corollary \ref{c:inter_disc}). 

Given $A$ and $\B$, there is a simple test to decide if $D(A \mid \B)$ is
finite, i.e., if there exists a finite sequence that generates $A$ from
$\B$. Let $\B = \{B_1, \ldots, B_m\}$. Given $w \in \amb$, we let
$\mathsf{vec}(w) \in \{0,1\}^m$ be the vector with $\mathsf{vec}(w)_i = 1$
if and only if $w \in B_i$. For a set $C \subseteq \amb$, let
$\mathsf{vec}(C) = \{\mathsf{vec}(c) \mid c \in C \}$. For vectors $u, v
\in \{0,1\}^n$, we write $u \preceq v$ if $u_i \leq v_i$ for each $i \in
[n]$.\footnote{
    We note that $\mathsf{vec}(w)$ always has Hamming weight exactly 2
    when $\calb = \calg_{N,M}$ and $w \in [N] \times [M]$.
    There is a well-known connection between slice functions and graph
    complexity (see, e.g.,~\cite{DBLP:journals/mst/Lokam03}).
}

\begin{proposition}[Finiteness test]\label{prop:finite}
$D(A \mid \B)$ is finite if and only if there are no vectors $u \in \mathsf{vec}(A)$ and $v \in \mathsf{vec}(A^c)$ such that $u \preceq v$.
\end{proposition}

\begin{proof}
Let $a \in A$ and $b \in A^c$ be elements such that $u = \mathsf{vec}(a) \preceq \mathsf{vec}(b) = v$. Suppose there is a construction $A_1, \ldots, A_t$ of $A$ from $\B$. It follows easily by induction that $b \in A_t$, which is contradictory. On the other hand, if there is no element $b$ and vector $v$ with this property, it is not hard to see that $A = \bigcup_{u \in \mathsf{vec}(A)} \bigcap_{i:u_i = 1} B_i$. This completes the proof of the proposition.
\end{proof}

Finally, observe that standard counting arguments yield the existence of sets of high discrete complexity.

\begin{lemma}[Complex sets]\label{l:complex_sets}
Let $k = |\amb|$ and $m = |\B|$. If $\;3s \lceil \log (m + s) \rceil < k$, there exists a set $A \subseteq \amb$ such that $D(A \mid \B) \geq s$. 
\end{lemma}

For instance, a random matrix $M \subseteq [N] \times [N]$ satisfies $D(M
\mid \R_{N,N}) = \Omega(N)$, while a random graph $G \subseteq [N] \times
[N]$ has $D(G \mid \G_{N,N}) = \Omega(N^2/\log N)$. It is easy to see that
the former lower bound is asymptotically tight. The tightness of the graph
complexity bound is also known (cf.~\citep[Theorem 1.7]{juknacomputational}).

\subsection{Transference of lower bounds}\label{sec:transference}

The following lemma generalizes a similar reduction from graph complexity (see, e.g., 
\citep[Section 1.3]{juknacomputational}).

\begin{lemma}\label{l:injective}
Let $\langle \amb_1, \mathcal{B}_1 \rangle$ and $\langle \amb_2, \mathcal{B}_2 \rangle$ be discrete spaces, and $\phi \colon \amb_1 \to \amb_2$ be an injective function. Assume that $\B_2 = \{B^2_1, \ldots, B^2_m\}$. Then, for every $A_1 \subseteq \amb_1$,
\begin{eqnarray}
D(\phi(A_1) \mid \mathcal{B}_2) & \geq & D(A_1 \mid \mathcal{B}_1) - D(\phi^{-1}(B^2_1), \ldots, \phi^{-1}(B^2_m) \mid \mathcal{B}_1) \nonumber \\
& \geq & D(A_1 \mid \mathcal{B}_1) - \sum_{B \in \mathcal{B}_2} D(\phi^{-1}(B) \mid \mathcal{B}_1). \nonumber
\end{eqnarray}
The result also holds with respect to the discrete complexity measures $D_\cap$ and $D_\cup$.
\end{lemma}

\begin{proof}
Let $A_2 = \phi(A_1)$. Since $\phi$ is injective, $\phi^{-1}(A_2) = A_1$. Let $B^2_1, \ldots, B^2_m, C_1, \ldots, C_t = A_2$ be an extended sequence that describes a construction of $A_2$ from $\B_2$, where $t = D(A_2 \mid \B_2)$. We claim that $$\phi^{-1}(B^2_1), \ldots, \phi^{-1}(B^2_m), \phi^{-1}(C_1), \ldots, \phi^{-1}(C_t) = A_1$$ is an extended sequence that describes a construction of $A_1$ from $\{\phi^{-1}(B^2_1), \ldots, \phi^{-1}(B^2_m)\}$. Indeed, this can be easily verified by induction using that $\phi^{-1}(C_1 \cap C_2) = \phi^{-1}(C_1) \cap \phi^{-1}(C_2)$ and $\phi^{-1}(C_1 \cup C_2) = \phi^{-1}(C_1) \cup \phi^{-1}(C_2)$. The result immediately follows by replacing the initial sets in the construction above by a sequence that realizes $D(\phi^{-1}(B^2_1), \ldots, \phi^{-1}(B^2_m) \mid \mathcal{B}_1)$.
\end{proof}

In particular, if we have a strong enough lower bound with respect to $\langle \amb_1, \B_1 \rangle$, and can construct an injective map $\phi \colon \amb_1 \to \amb_2$ such that for each $B \in \B_2$ the value $D(\phi^{-1}(B) \mid \B_1)$ is small, we get a lower bound in $\langle \amb_2, \B_2 \rangle$. Moreover, if the original set $A_1$ and the map $\phi$ are ``explicit'',  $A_2 = \phi(A_1)$ is explicit as well.

We provide next a simple example that will be useful later in the text.
Given a binary string $w \in \{0,1\}^n$, which we represent as $w = w_1
\ldots w_n$, let $\numb(w) = \sum_{i=0}^{n-1} 2^i \cdot w_{n -
i}$ be the number in $\{0, \ldots, 2^n -1\}$ encoded by $w$. Let $N = 2^n$,
and let 
$\binary \colon [N] \to \{0,1\}^n$ 
be the \emph{bijection} that maps
the integer $\numb(w) + 1$ to the corresponding string $w \in
\{0,1\}^n$.

\begin{lemma}[Tight transference from graph complexity to circuit complexity]\label{l:graph_to_circuit}
Let $\langle [N] \times [N], \G_{N,N} \rangle$ and $\langle \{0,1\}^{2n}, \B_{2n} \rangle$ be the discrete spaces corresponding to $N \times N$ graph complexity and $2n$-bit circuit complexity, respectively, where $N = 2^n$. Moreover, let $\phi \colon [N] \times [N] \to \{0,1\}^{2n}$ be the bijective map defined by $\phi(u,v) \eqdef \binary(u)\binary(v)$. For every $G \subseteq [N] \times [N]$,
$$
D_\cap(\phi(G) \mid \B_{2n}) \;\geq\; D_\cap(G \mid \G_{N,N}). 
$$
In particular, graph intersection complexity lower bounds yield circuit complexity lower bounds.
\end{lemma}

\begin{proof}
By Lemma \ref{l:injective}, it is enough to verify that for each $B \in
\B_{2n}$, $D_\cap(\phi^{-1}(B) \mid \G_{N,N}) = 0$. Recall from Section
\ref{ss:circuit_complexity} that $\B_{2n} = \{B_1, \ldots, B_{2n}, B_1^c,
\ldots, B_{2n}^c\}$, where $B_i = \{ v \in \{0,1\}^{2n} \mid v_i = 1\}$. If
$B_i \in \B_{2n}$ corresponds to the positive literal $x_i$, then
$\phi^{-1}(B_i)$ is either a union of columns (when $i>n$) or a union of
rows (when $i \leq n$) in graph complexity (cf.~Section
\ref{ss:graph_complexity}). Consequently, in this case $D_\cap(\phi^{-1}(B_i)
\mid \G_{N,N}) = 0$ by Facts \ref{f:first_trivial_fact} and
\ref{f:chain_rule}. On the other hand, for a $B_i^c \in \B_{2n}$, it is not
hard to see that $\phi^{-1}(B_i^c)$ also corresponds to either a union of
rows or a union of columns. This completes the proof.
\end{proof}

An advantage of Lemma \ref{l:graph_to_circuit} over existing results
connecting graph complexity and circuit complexity is that it offers a
tighter connection between these two models by focusing on a  convenient
complexity measure (intersection complexity instead of circuit
complexity).\footnote{
    In the Magnification Lemma of~\cite{juknacomputational},
    it is already implicitly shown that 
    $D_\cap(f_G \mid \calb_{2n}) \geq D_{\cap}(G \mid \calg_{N,N})$.
    However, 
    the literature in graph complexity focuses on the relationship between
    $D(f_G \mid \calb_{2n})$ and $D(G \mid \calg_{N,N})$,
    where there is a constant factor loss.
    In particular,
    the best transference bound known is
    $D(f_G \mid \calb_{2n}) \geq  D(G \mid \calg_{N,N}) - (4+o(1))N$
    (see~\cite{juknacomputational}, citing~\cite{CHASHKIN+1994+229+258}).
    This means that only a $\Omega(N)$
    lower bound
    on
    $D(G \mid \calg_{N,N})$
    would imply a meaningful bound on $D(f_G \mid \calb_{2n})$,
    whereas our setting allows us to transfer a
    $(1+\eps)\log N$ graph complexity lower bound into a $(1+\eps)n$ circuit lower bound.
}

\begin{remark}[Circuit lower bounds from graph complexity lower bounds] 
    \label{rk:transf}
Let $C \geq 1$ be a constant. We note that a lower bound of the form $C \cdot \log N$ on $D_\cap(H \mid \G_{N,N})$ for an explicit graph $H$ can be translated into the same lower bound on the circuit complexity of a related explicit Boolean function. In more detail, let $f_H \colon \{0,1\}^{2n} \to \{0,1\}$ be the Boolean function corresponding to a bipartite graph $H \subseteq [N] \times [N]$. Now consider the function $F \colon \{0,1\}^{1 + 2n} \to \{0,1\}$ defined as follows. The value $F(b,z) = f_H(z)$ if the input bit $b = 1$, and $F(b,z) = \overline{f_H}(z) = 1 - f_H(z)$ if $b = 0$. Note that if $H$ can be computed in time $\mathsf{poly}(N)$ then the corresponding function $F$ is in $E = \mathsf{DTIME}[2^{O(m)}]$, where $m = 2n + 1$ is the input length of $F$. Moreover, if $D_\cap(H \mid \G_{N,N}) \geq C \cdot \log N$ then any Boolean circuit computing $F$ must contain at least $C \cdot 2n$ $\mathsf{AND}$ and $\mathsf{OR}$ gates in total \emph{(}assuming a circuit model with access to input literals and without $\mathsf{NOT}$ gates\emph{)}. This follows from \Cref{l:graph_to_circuit} and Boolean duality, i.e., that the $\mathsf{AND}$ complexity of a Boolean function coincides with the $\mathsf{OR}$ complexity of its negation. Formally, letting $\B_\ell$ denote the standard set of generators in the Boolean circuit complexity of $\ell$-bit Boolean functions, we have:
\begin{eqnarray}
    D(F \mid \B_{m}) & \geq & D_\cap(F \mid \B_{m}) + D_\cup(F \mid \B_{m}) \nonumber \\
    & \geq & D_\cap(f_H \mid \B_{m}) + D_\cup(\overline{f_H} \mid \B_{m}) - O(1) \nonumber \\
     & = & D_\cap(f_H \mid \B_{m}) + D_\cap(f_H \mid \B_{m}) - O(1) \nonumber \\
     & \geq & 2 \cdot D_\cap(H \mid \G_{N, N}) - O(1) \nonumber \\ 
     & \geq &  2 \cdot C \cdot \log N = C \cdot 2n = C \cdot m - O(1). \nonumber
\end{eqnarray}
\end{remark}

\begin{remark}[Graph complexity lower bounds from circuit complexity lower bounds]
It is not hard to show by \Cref{l:injective} and a similar argument that a lower bound of the form $\omega(2^n \cdot n)$ on the circuit complexity of a function $h \colon \{0,1\}^{2n} \to \{0,1\}$ implies a $\omega(N)$ lower bound in graph complexity, where $N = 2^n$ as usual. On the other hand, note that by a counting argument there exist graphs computed by a single \emph{(}unbounded fan-in\emph{)} union whose corresponding $2n$-bit Boolean function has circuit complexity $\Omega(2^n/n)$. In particular, it follows from Lemma \ref{l:zwick} that a Boolean function can have exponential intersection complexity, while the corresponding graph has zero intersection complexity.
\end{remark}

\subsection{Cyclic Discrete Complexity}\label{ss:cyclic}

We introduce a variant of the complexity measure $D(\cdot \mid \cdot)$ that
allows cyclic constructions. Formally, we use $D^\circ(A \mid \B)$ to
denote the \emph{cyclic discrete complexity} of $A$ with respect to $\B$,
defined as follows. We consider a \emph{syntactic sequence} $I_1, \ldots,
I_t$, together with a fixed operation of the form $I_i = K_{i_1} \star_i
K_{i_2}$, where $K_{i_1}, K_{i_2} \in \{I_1, \ldots, I_t\} \cup \B$ and
$\star_i \in \{\cap, \cup\}$, for each $i \in [t]$. (Notice that we do not
require $i_1, i_2 < i$.) The syntactic sequence is viewed as a formal
description instead of an actual construction, and it is evaluated as
follows. Initially, $I_i^0 \eqdef \emptyset$ for each $i \in [t]$. Then,
for every $j > 0$, $I^j_i \eqdef I^{j-1} \cup (K^{j-1}_{i_1} \star_i
K^{j-1}_{i_2})$, where the sets in $\B$ remain fixed throughout the
evaluation. We say that the syntactic sequence generates $A$ from $\B$ if
there exists $j \in \mathbb{N}$ such that $I_t^{j'} = A$ for every $j' \geq
j$. Finally, we let $D^\circ(A \mid \B)$ denote the minimum length $t$ of
such a sequence, if it exists. The complexity measure $D^\circ_\cap$ is
defined analogously, and only takes into account the number of intersection
operations in the definition of the syntactic sequence.

\begin{lemma}[Convergence of the evaluation procedure]\label{l:convergence}
Suppose $I_1, \ldots, I_t$ together with the corresponding $\star_i$ operations define a syntactic sequence. Then, for every $j \geq t$, $$I^{j + 1}_i \;=\; I^{j}_i.$$ In other words, the evaluation converges after at most $t$ steps.
\end{lemma}

\begin{proof}
The evaluation is monotone, in the sense that an element $v \in \amb$ added to a set during the $j$-th step of the evaluation cannot be removed in subsequent updates. From the point of view of this fixed element, if it is not added to a new set during an update, it won't be added to new sets in subsequent updates. Consequently, each set in the sequence converges after at most $t$ iterations. 
\end{proof}

\begin{corollary}[Cyclic discrete complexity versus discrete complexity]\label{c:circ_versus_disc}
For every set $A \subseteq \amb$ and family $\B \subseteq \mathcal{P}(\amb)$ of generators, $$D^\circ_\cap(A \mid \B) \;\leq\; D_\cap(A \mid \B) \;\leq\; D^\circ_\cap(A \mid \B)^2.$$
\end{corollary}

\begin{proof}
For the first inequality, observe that from every construction of $A$ from
$\B$ we can define an acyclic syntactic sequence that generates $A$ from
$\B$. For the second inequality, simply unfold the evaluation of the
syntactic sequence into a sequence that generates $A$ from $\B$. Since the
additional union operations coming from the update step $I^j_i = I^{j-1}
\cup (K^{j-1}_{i_1} \star_i K^{j-1}_{i_2})$ do not increase intersection
complexity, the claimed upper bound follows from
\Cref{l:convergence}.
\end{proof}

We will employ cyclic discrete complexity in Section \ref{ss:fusion_cyclic} to exactly characterize the power of the fusion method as a framework to lower bound discrete complexity. We finish this section with a concrete example that is relevant in the context of the fusion method (cf.~Section \ref{ss:fusion_upper}).\\

\noindent \textbf{Example: The Fusion Problem $\Pi_\mathcal{R}$.} Let $[m]
= \{1, \ldots, m\}$, $Y \subseteq [m]$ be 
a subset of $[m]$, and
$\mathcal{R}$ be a \emph{fixed} set of rules encoded by a set of triples of
the form $(a,b,c)$, where $a,b,c \in [m]$ are arbitrary. The meaning of a
rule $(a,b,c)$ is that the element $c$ should be added to $Y$ in case this
set already contains elements $a$ and $b$. We let $\Pi_\mathcal{R}$ be the
following computational problem: Given an arbitrary initial set $Y
\subseteq [m]$ as an input instance, is the top element $m$ eventually
added to $Y$? (Observe that this problem is closely related to the GEN
Boolean function investigated in \citep{DBLP:journals/combinatorica/RazM99}
and related works.)

Note that, for every fixed set $\mathcal{R}$ of rules, $\Pi_\mathcal{R}$ can be decided by a cyclic monotone Boolean circuit that contains exactly $|\mathcal{R}|$ fan-in two AND gates. Indeed, it is enough to consider a circuit over input variables $y_1, \ldots, y_m$ that contains three additional layers of gates, described as follows. The first layer contains fan-in two OR gates $f_1, \ldots, f_m$, where each $f_i$ is fed by the input variable $y_i$ and by a corresponding gate $h_i$ in the third layer. Each rule $(a,b,c) \in \mathcal{R}$ gives rise to a fan-in two AND gate $g_{a,b,c}$ in the second layer of the circuit, where $g_{a,b,c} = f_a \wedge f_b$. Finally, in the third layer we have for each $i \in [m]$ a corresponding OR gate $h_i$, where $$h_i = \bigvee_{u,v \in [m], (u,v,i) \in \mathcal{R}} g_{u,v,i}.$$ (We stress that \emph{unbounded fan-in} gates are used only to simplify the description of the circuit.) It is easy to see that the gate $f_m$ computes $\Pi_\mathcal{R}$ after at most $O(|\mathcal{R}|)$ iterations of the evaluation procedure. 

\section{Characterizations of Discrete Complexity via Set-Theoretic Fusion}\label{s:fusion}

The technique presented in this section can be seen as a set-theoretic formulation of some results from \citep{DBLP:conf/stoc/Razborov89} and \citep{DBLP:conf/coco/Karchmer93}. The tighter characterization that appears in Section \ref{ss:fusion_cyclic} is an adaptation of a result from \citep{DBLP:journals/iandc/NakayamaM95}.

\subsection{Definitions and notation}\label{ss:fusion_defin_notat}

For convenience, let $U \eqdef A^c = \amb \setminus A$, where $\amb$ is the ambient space. We assume from now on that $A$ is \emph{non-trivial}, i.e., both $A$ and $A^c$ are non-empty.  

\begin{definition}[Semi-filter]
We say that a non-empty family $\F \subseteq \mathcal{P}(U)$ of sets is a \emph{semi-filter} over $U$ if the following hold:
\begin{itemize}
\item \emph{(upward closure)} If $U_1 \in \F$ and $U_1 \subseteq U_2 \subseteq U$, then $U_2 \in \F$.
\item \emph{(non-trivial)} $\emptyset \notin \F$.
\end{itemize}
\end{definition}

\begin{definition}[Semi-filter above $w$]\label{d:filter_above}
We say that $\F$ is \emph{above} an element $w \in \amb$ \emph{(with respect to $\B$ and $U = A^c$)} if the following condition holds. For every $B \in \B$, if $w \in B$ then $B_U \in \F$.
\end{definition}

\Cref{fig:example} illustrates \Cref{d:filter_above}
in the particularly simple and attractive 2-dimensional framework of graph complexity 
considered in this work.

\setcounter{MaxMatrixCols}{30}

\begin{figure}[H]
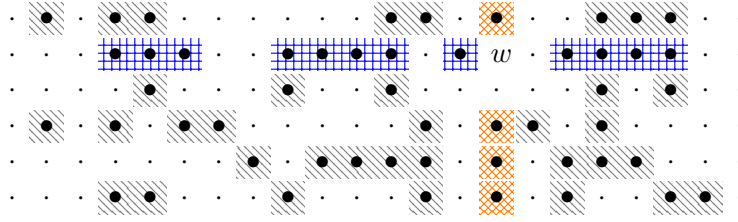

\begin{equation*}
\begin{matrix}

\cdot & \hlc{\bullet}		&	\cdot		&	\hlc{\bullet}&	\hlc{\bullet}&	\cdot		&	\cdot		&	\cdot		& \cdot & \cdot		&	\cdot		&	\hlc{\bullet}&	\hlc{\bullet}&	\cdot		&	\hlb{\bullet}		&	\cdot		& \cdot &	\hlc{\bullet} &	\hlc{\bullet}&	\hlc{\bullet}&	\cdot & \cdot\\

\cdot & \cdot		&	\cdot		&	\hlg{\bullet}&	\hlg{\bullet}&	\hlg{\bullet}		&	\cdot		&	\cdot		& \hlg{\bullet} & \hlg{\bullet}&	\hlg{\bullet}&	\hlg{\bullet}&	\cdot		&	\hlg{\bullet}&	w	\hspace{-0.13cm}		&	\cdot		& \hlg{\bullet} &	\hlg{\bullet} &	\hlg{\bullet}&	\hlg{\bullet}&	\cdot & \cdot\\

\cdot & \cdot		&	\cdot		&	\cdot &	\hlc{\bullet}&	\cdot		&	\cdot		&	\cdot		& \hlc{\bullet} &\cdot		&	\cdot		&	\hlc{\bullet}&	\cdot&	\cdot		&	\cdot		&	\cdot&\cdot &	\hlc{\bullet} &	\cdot&	\hlc{\bullet}&	\cdot & \cdot\\

\cdot & \hlc{\bullet}		&	\cdot		&	\hlc{\bullet}&	\cdot&	\hlc{\bullet}		&	\hlc{\bullet}		&	\cdot		& \cdot &\cdot		&	\cdot		&	\cdot		&	\hlc{\bullet}&	\cdot		&	\hlb{\bullet}&	\hlc{\bullet}&\cdot &	\hlc{\bullet} & \cdot &	\cdot&	\cdot & \cdot\\

\cdot & \cdot		&	\cdot		&	\cdot&	\cdot&	\cdot		&	\cdot		&	\hlc{\bullet}		& \cdot& \hlc{\bullet}		&	\hlc{\bullet}		&	\hlc{\bullet}		&	\hlc{\bullet}&	\cdot		&	\hlb{\bullet}&	\cdot		&\hlc{\bullet}  &	\hlc{\bullet} &	\hlc{\bullet}&	\cdot&	\cdot & \cdot\\

\cdot & \cdot		&	\cdot		&	\hlc{\bullet}&	\hlc{\bullet}&	\cdot		&	\cdot		&	\cdot		& \hlc{\bullet} & \cdot		&	\cdot		&	\cdot		&	\hlc{\bullet}&	\cdot		&	\hlb{\bullet}&	\cdot		&\hlc{\bullet}  &	\cdot &	\cdot&	\hlc{\bullet}&	\hlc{\bullet} & \cdot

\end{matrix}
\end{equation*}
\caption{In this example, $\amb = [6] \times [22]$, $\B = \G_{6, 22}$ (as in Section \ref{ss:graph_complexity}), and the $\{\cdot, \bullet, w\}$-valued matrix above encodes $U = G^c$ (rectangles with $\bullet$), where $G \subseteq \amb$ (locations with $\cdot$ and $w$) can be interpreted as a bipartite graph. If a semi-filter $\F$ over $U$ is above $w \in G$ (corresponding to coordinates $(2,15)$), then it must contain the distinguished subsets of $U$ represented in blue $(R_2 \cap U)$ and in orange $(C_{15} \cap U)$, respectively.}
\label{fig:example}
\end{figure}

Intuitively, semi-filters will be used to produce counter-examples to the correctness of a candidate construction of a set $A$ from $\B$ that is more efficient than $D_\cap(A \mid \B)$. This will become clear in Section \ref{ss:fusion_lower}.

\begin{definition}[Preservation of pairs of subsets]
Let $\Lambda = \{(E_1,H_1), \ldots, (E_\ell,H_\ell)\}$ be a family of pairs of subsets of $U$. We say that $\F$ \emph{preserves} a pair $(E_i, H_i)$ if $E_i \in \F$ and $H_i \in \F$ imply $E_i \cap H_i \in \F$. We say that $\F$ \emph{preserves} $\Lambda$ if it preserves every pair in $\Lambda$.
\end{definition}

We now introduce a measure of the \emph{cover complexity} of $A \subseteq \amb$ with respect to a family $\B \subseteq \mathcal{P}(\amb)$.

\begin{definition}[Cover complexity]
We let $\rho(A, \B) \in \mathbb{N} \cup \{\infty\}$ be the minimum size of a collection $\Lambda$ of pairs of subsets of $U$ such that there is no semi-filter $\F$ over $U$ that preserves $\Lambda$ and is above an element $a \in A$ \emph{(}with respect to $\B$ and $U$\emph{)}.
\end{definition}

The definition of cover complexity considered here is with respect to
semi-filters (essentially, monotone functions which are not equal to the
constant function which outputs 1).
In the context of circuit complexity, notions of cover
complexity with respect to other types of Boolean functions (such as
ultrafilters and linear functions) have been considered, yielding characterizations of
different circuit models~\cite{Wigderson93thefusion}.
If we ask that in every pair at least one of the sets is the intersection
of a generator with $U$, we obtain characterizations of branching
models~\cite{wigderson1995lectures} (such as branching programs).
In \Cref{ss:nondet_graph_complexity},
we will consider the 2-dimensional cover problem with ultrafilters.

\paragraph{Cover Graph of $A$ and $\mathcal{B}$.} In order to get more intuition about the notion of cover complexity, consider an undirected bipartite graph $\Phi_{A, \B} = (V_\mathsf{pairs}, V_\mathsf{filters}, \mathcal{E})$, where 
\begin{eqnarray}
    V_\mathsf{pairs} & \eqdef & \{(E,H) \mid E, H \subseteq U\},   \nonumber \\
    V_\mathsf{filters} & \eqdef & \{\F \subseteq \mathcal{P}(U) \mid \F~\text{is a semi-filter and}~\F~\text{is above some}~a \in A\}, \nonumber
\end{eqnarray}
and there is an edge $e \in \mathcal{E}$ connecting $(E,H) \in V_\mathsf{pairs}$ and $\F \in V_\mathsf{filters}$ if and only if $\F$ does not preserve $(E,H)$. Then $\rho(A, \B)$ is the minimum number of vertices in $V_\mathsf{pairs}$ whose adjacent edges cover all the vertices in $V_\mathsf{filters}$. For convenience, we say that $\Phi_{A, \B}$ is the \emph{cover graph} of $A$ and $\B$.

Note that a set of vertices in $V_{\mathsf{pairs}}$ whose adjacent edges
cover all of the vertices in $V_{\mathsf{filters}}$ is also known as a
\emph{dominating set} in graph theory.
Moreover, identifying vertices with their neighbourhoods,
the value of $\rho(A, \mathcal{B})$ is equivalent
to the optimal value of a set cover problem.

\subsection{Discrete complexity lower bounds using the fusion method}\label{ss:fusion_lower}

\begin{theorem}[Fusion lower bound]\label{t:fusion_lower}
Let $A \subseteq \amb$ be non-trivial, and $\B \subseteq \mathcal{P}(\amb)$ be a non-empty family of generators. Then
$$\rho(A, \B) \; \leq \; D_\cap(A \mid \mathcal{B}).$$
In other words, the cover complexity of a non-trivial set lower bounds its intersection complexity.
\end{theorem}

Before proving the result, it is instructive to consider an example. Assume $\Gamma = [N] \times [N]$ and $\B = \mathcal{R}_N$ are instantiated as in Section \ref{ss:comb_rect}, where we noted that $D_\cap(G \mid \R_N)$ is always zero. Indeed, $\rho(G, \R_N) = 0$ for every non-trivial $G \subseteq [N] \times [N]$, since in the corresponding cover graph $\Phi_{G, \R_N}$ the vertex set $V_\mathsf{filters}$ is empty (observe that if a semi-filter is above some $a \in G$, then it must contain the empty set, which is contradictory).

\begin{proof}
Let $|\B| = m$ and $D_\cap(A \mid \B) = k$. Assume toward a contradiction that $k < \rho(A, \B)$. Let
\begin{equation}\label{eq:constr1}
C^1, \ldots, C^m, C^{m+1}, \ldots, C^{m +t} = A
\end{equation}
be an extended sequence of complexity $t$ that generates $A$ from $\B$, and suppose it has intersection complexity $k$. Let $U \eqdef A^c = \amb \setminus A$. Recall that, by assumption, both $A$ and $U$ are non-empty. Consider the corresponding relativized sequence
\begin{equation}\label{eq:constr2}
C^1_U, \ldots, C^m_U, C^{m+1}_U, \ldots, C^{m +t}_U = \emptyset.
\end{equation}
This extended sequence generates the empty set from $\B_U$ and has intersection complexity $k$. Let $\Lambda$ be the set of intersection operations in this sequence. Note that each pair $(C^i_U, C^j_U) \in \Lambda$ satisfies $C^i_U, C^j_U \subseteq U$, and that $|\Lambda| \leq k < \rho(A, \B)$. Let $\Phi_{A, \B} = (V_{\mathsf{pairs}}, V_{\mathsf{filters}}, \mathcal{E})$ be the cover graph of $A$ and $\B$. Since $\Lambda \subseteq V_{\mathsf{pairs}}$ and $|\Lambda| < \rho(A, \B)$, there exists $\F \in V_\mathsf{filters}$ that is not covered by the pairs in $\Lambda$. Let $a \in A$ be an element such that $\F$ is above $a$.

We trace the construction in Equation \ref{eq:constr1} from the point of
view of the element $a$. Let $\alpha_i = 1$ if and only if $a \in C_i$.
Observe that $\alpha_{m + t} = 1$, since $a \in A$. In order to derive a
contradiction, we define a second sequence $\beta_i$ that depends on the
semi-filter $\F$ and on the relativized construction appearing in  Equation
\ref{eq:constr2}. We let $\beta_i = 1$ if and only if $C^i_U \in \F$
(recall that $\F \subseteq \mathcal{P}(U)$ and $C^i_U \subseteq U$). Since
$\F$ is a semi-filter and $C^{m + t}_U = \emptyset$, we get $\beta_{m + t}
= 0$. We complete the argument by showing that for every $i \in [m + t]$,
$\alpha_i \leq \beta_i$, which is in contradiction to $\alpha_{m + t} = 1$
and $\beta_{m + t} = 0$.

\begin{claim}\label{c:majorize}
Suppose $\F$ is above $a \in A$ with respect to $\B$ and $U$, and that $\F$ preserves $\Lambda$, the set of intersection operations in Equation \ref{eq:constr2}. Then for every $i \in [m + t]$, $\alpha_i \leq \beta_i$.
\end{claim}

The proof is by induction on $i$. For the base case, we consider $i \leq m$. Since $\B$ is non-empty, $m \geq 1$. Now if $\alpha_i = 1$, then $a \in C^i = B$ for some $B \in \B$. Since $\F$ is above $a$ (with respect to $\B$ and $U$) and $a \in B$, $C^i_U = B_U \in \F$, and consequently $\beta_i = 1$. This completes the base case. 

The induction step follows from the induction hypothesis and the upward
closure of $\F$ in the case of a union operation, and from the induction
hypothesis and the fact that $\F$ preserves $\Lambda$ in the case of an
intersection operation. For instance, suppose that $C^i = C^{i_1} \cap
C^{i_2}$ and $C^i_U = C^{i_1}_U \cap C^{i_2}_U$, respectively, where $i_1,
i_2 < i$. Assume that $\alpha_i = 1$. Then $a \in C^i$, and consequently $a
\in C^{i_1} \cap C^{i_2}$. Using the induction hypothesis, $1 =
\alpha_{i_1} = \alpha_{i_2} = \beta_{i_1} = \beta_{i_2}$. Therefore,
$C^{i_1}_U \in \F$ and $C^{i_2}_U \in F$. Now using that $(C^{i_1}_U,
C^{i_2}_U) \in \Lambda$ and that $\F$ preserves $\Lambda$, it follows that
$C^i_U = C^{i_1}_U \cap C^{i_2}_U \in F$. In other words, $\beta_i = 1$.
The other case is similar.\\

This establishes the claim and 
completes the proof of Theorem \ref{t:fusion_lower}.
\end{proof}

\subsection{Set-theoretic fusion as a complete framework for lower bounds}\label{ss:fusion_upper}

In this section, we establish a converse to Theorem \ref{t:fusion_lower}.

\begin{theorem}[Fusion upper bound]\label{t:fusion_upper}
Let $A \subseteq \amb$ be non-trivial, and $\B \subseteq \mathcal{P}(\amb)$ be a non-empty family of generators. Then
$$
D_\cap(A \mid \mathcal{B}) \;\leq \; \rho(A, \B)^2.
$$
\end{theorem}

\begin{remark}
It is important in the statements of Theorems \ref{t:fusion_lower} and \ref{t:fusion_upper} that the characterization of $D_\cap(A \mid \B)$ in terms of $\rho(A, \B)$ does not suffer a quantitative loss that depends on $|\mathcal{B}|$. This allows us to apply the results in discrete spaces for which the number of generators in $\B$ is large compared to the size of the ambient space $\Gamma$, such as in graph complexity.
\end{remark}

\begin{proof}
Let $U = A^c$, let $\rho(A,\B) = t$, and assume that this is witnessed by a family $$\Lambda = \{(H_1, E_1), \ldots, (H_t,E_t)\}$$ of $t$ pairs of subsets of $U$. We let
$$
\mathfrak{F}_\Lambda = \{ \F \subseteq \mathcal{P}(U) \mid \F~\text{is a semi-filter that preserves}~\Lambda\}.
$$
Recall the definition of the cover graph $\Phi_{A, \B}$ of $A$ and $\B$ (Section \ref{ss:fusion_defin_notat}). Observe that, while $\Lambda \subseteq V_\mathsf{pairs}$, it is not necessarily the case that $\mathfrak{F}_\Lambda \subseteq V_\mathsf{filters}$.

\begin{claim}\label{cl:cond1}
For every $w \in \amb$, 
$$
w \in A \quad \text{if and only if} \quad \nexists \F \in \mathfrak{F}_\Lambda~\text{that is above}~w~(\text{w.r.t.}~\B~\text{and}~U).
$$
\end{claim}
In order to see this, notice that if $w \in A$ then indeed there is no such
$\F \in \mathfrak{F}_\Lambda$, using the definitions of $\rho$ and
$\Lambda$. On the other hand, for $w \notin A$, it is easy to check that
$\F_w \eqdef \{U' \subseteq U \mid w \in U'\}$ is a semi-filter that
preserves $\Lambda$ and that is above $w$ with respect to $\B$ and $U$.\\

This claim provides a criterion to determine if an element is in $A$. This will be used in a construction of $A$ from $\B$ showing that $D_\cap(A \mid \B) = O(\rho(A, \B)^2)$. The intuition is that, for a given $w \in \amb$, we must check if there is $\F \in \mathfrak{F}_\Lambda$ that is above $w$ with respect to $\B$ and $U$. In order to achieve this, we inspect the \emph{minimal family} $\G_w \subseteq \mathcal{P}(U)$ of sets that must be contained in any such (candidate) semi-filter. 

For every $w \in \amb$, we require $\G_w$ to be above $w$, upward-closed,
and to preserve $\Lambda$. The rules for constructing $\G_w$ are simple:

\begin{itemize}
\item \emph{Base case}. If $w \in B$ for $B \in \B$, then add $B_U = B \cap U$ to $\G_w$, together with every set $U'$ such that $B_U \subseteq U' \subseteq U$. 
\item \emph{Propagation step}. If both $E_i$ and $H_i$ are in $\G_w$, add
    $E_i \cap H_i$ to $\G_w$, together with every set $U'$ such that $E_i
    \cap H_i \subseteq U' \subseteq U$.
\end{itemize}

\noindent We apply the base case once, and repeatedly invoke the propagation step until no new sets are added to $\G_w$. Clearly, this process terminates within a finite number of steps.

\begin{claim}\label{cl:cond2}
For every $w \in \amb$, the empty set is added to $\G_w$ if and only if $w \in A$.
\end{claim}

We argue that $w \notin A$ if and only if $\emptyset \notin \G_w$. Clearly,
if $\F$ is a semi-filter that is above $w$ and preserves $\Lambda$, we must
have $\G_w \subseteq \F$. For $w \notin A$, the process described above
cannot possibly add $\emptyset$ to $\G_w$, since by Claim \ref{cl:cond1}
there is a semi-filter $\F \in \mathfrak{F}_\Lambda$ that is above $w$, and
$\G_w \subseteq \F$. On the other hand, if this process terminates without
adding $\emptyset$ to $\G_w$, it is easy to see that $\G_w$ is a
semi-filter in $\mathfrak{F}_\Lambda$ that is above $w$, which in turn
implies that $w \notin A$ via Claim \ref{cl:cond1}. This completes the
proof of Claim \ref{cl:cond2}.\\

We now turn this discussion into the actual construction of $A$ from the sets in $\B$. For convenience, we actually upper bound $D_\cap(A \mid \B \cup \{\emptyset\})$, i.e., we freely use $\emptyset$ as a generator in the description of the sequence that generates $A$. This is without loss of generality due to Fact \ref{f:use_empty}. Let
$$
\Omega \;\eqdef\; \B_U \cup \{E_i\}_{i \in [t]} \cup \{H_i\}_{i \in [t]} \cup \{H_i \cap E_i\}_{i \in [t]} \cup \{\emptyset\}, 
$$
where we abuse notation and view $\Omega$ as a \emph{multi-set}.\footnote{This is helpful in the argument. For instance, more than one set $B \in \B$ might generate an empty set $B_U = B \cap U \in \Omega$, but we will need to keep track of elements such that $w \in B$ and $B_U = \emptyset$.} For simplicity and in order to avoid extra terminology, we slightly abuse notation, and distinguish sets that are identical by the symbols representing them. This should be clear in each context, and the reader should keep in mind that we are simply translating the process that defines each $\G_w$ into a construction of $A$. 

Fix a set $C$ from the multi-set $\Omega$. For an integer $j\geq 1$, we let $S^j_C$ be the set of all $w \in \amb$ that have $C$ in $\G_w$ before the start of the $j$-th iteration (propagation step) of the process described above. (Here we also view the sets $S^j_C$ as different formal objects.) We construct each set $S^j_C$ from $\B \cup \{\emptyset\}$ by induction on $j$. By Claim \ref{cl:cond2}, for a large enough $\ell \in \mathbb{N}$, we get $S^\ell_\emptyset = A$, our final goal.

In the base case, i.e., for $j=1$, we first set $T^1_{B_U} = B$ for each
$B_U$ obtained from a set $B \in \B$, and $T^1_I = \emptyset$ for every
other set $I \in \Omega$. We then let
\begin{equation} \label{eq:num1}
S^1_C \;=\; \bigcup_{C' \in \Omega, C' \subseteq C} T^{1}_{C'},
\end{equation} 
for each $C \in \Omega$. Observe that the base case satisfies the property in the definition of the sets $S^j_C$.

Assume we have constructed $S^{j-1}_C$, for each $C \in \Omega$. We can construct each $S^j_C$ from these sets as follows:
\begin{eqnarray}
T^j_C & = & S^{j-1}_C \cup \bigcup_{\{i \in [t] \;\mid\; C = E_i \cap H_i\}} (S^{j-1}_{E_i} \cap S^{j-1}_{H_i}), \label{eq:num2} \\
S^j_C & = & \bigcup_{C' \in \Omega, C' \subseteq C} T^{j}_{C'}. \label{eq:num3}
\end{eqnarray}
Note that the definition of each $S^j_C$ handles $\Lambda$-preservation and upward-closure, as in the propagation step. It is not difficult to show using the induction hypothesis that each set $S^j_C$ satisfies the required property (fix an element $w \in \amb$, and verify that it appears in the correct sets). This completes the construction of $A$. 

In order to finish the proof of Theorem \ref{t:fusion_upper}, we analyse the complexity of this construction. First, since each propagation step that introduces a new set to $\G_w$ adds at least one of the sets $E_i \cap H_i$ to $\G_w$, and there are at most $t = |\Lambda| = \rho(A, \B)$ such sets, it is sufficient in the construction above to take $\ell = t + 1$. In particular, $S^{t+1}_\emptyset = A$. Finally, each propagation step (which is associated to a fixed stage $j \in [t]$ of the construction) only employs intersection operations for sets $C$ of the form $E_i \cap H_i$ (in the corresponding definition of $T^i_C$). Overall, among these sets, the $j$-th stage of the construction needs at most $t$ intersections. To see this, note that sets $S^j_{C}$ with $C = E_i \cap H_i$ are only required to inspect the corresponding sets associated with pairs $(E_k, H_k)$ with $k \in [t]$ such that $C = E_k \cap H_k$, and such pairs are disjoint among the different sets $C$ of this form. (There is no need to keep more than one such $C$ representing the same underlying set as a syntactical object in the construction.) 

This immediately implies that $A$ can be generated using at most $t(t+1)$
intersections. 
However, since intersections are only added in steps $j \in
\set{2,\dots,t+1}$,
we obtain
$D_\cap(A \mid \B) \leq \rho(A, \B)^2$, which completes the proof. 
\end{proof}

We take this opportunity to observe the following immediate consequence of Theorems \ref{t:fusion_lower} and \ref{t:fusion_upper}. (A tighter relation between these measures is discussed in Section \ref{ss:basic_results}.)

\begin{corollary}[Intersection complexity versus discrete complexity] \label{c:inter_disc}~\\
For every $A \subseteq \amb$ and non-empty $\B$, if $D_\cap(A \mid \B) = t$ then $D_\cup(A \mid \B) \leq D(A \mid \B) \leq O(t + |\B|)^3$.
\end{corollary}

\begin{proof}
If $A$ is empty and can be constructed from $\B$, then it can also be constructed from $\B$ using $|\B|$ intersections (and no union operation). If $A = \amb$ the same is true with respect to unions. On the other hand, for a non-trivial $A$, the result follows from Theorems \ref{t:fusion_lower} and \ref{t:fusion_upper}, by noticing that in the construction underlying the proof of Theorem \ref{t:fusion_upper} a total of at most $O(t + |\B|)^3$ operations are needed.
\end{proof}

\begin{remark}[The fusion method and complexity in Boolean algebras]
Our presentation allows us to conclude, in particular, that the fusion method provides a framework to lower bound the number of operations in any \emph{(}finite\emph{)} Boolean algebra $\mathfrak{B}$. Indeed, by the Stone Representation Theorem \emph{(}cf.~\emph{\citep{givant2008introduction})}, any Boolean algebra is isomorphic to a field of sets. Therefore, the problem of determining the number of $\vee_{\B}$ and $\wedge_{\B}$ operations necessary to obtain a non-trivial element $a \in \mathfrak{B}$ from elements $b_1, \ldots, b_m \in \mathfrak{B}$ can be captured via cover complexity by Theorems \ref{t:fusion_lower} and \ref{t:fusion_upper}.
\end{remark}

\subsection{An exact characterization via cyclic discrete complexity}\label{ss:fusion_cyclic}

In this section, we show that cover complexity can be \emph{exactly} characterized using the intersection complexity variant of cyclic complexity. The tight correspondence is obtained by a simple adaptation of an idea from \citep{DBLP:journals/iandc/NakayamaM95}.

\begin{theorem}[Exact characterization of cover complexity]\label{t:fusion_cyclic}
Let $A \subseteq \amb$ be non-trivial, and $\B \subseteq \mathcal{P}(\amb)$ be a non-empty family of generators. Then
$$
\rho(A, \B) \;= \; D_\cap^\circ(A \mid \mathcal{B}).
$$
\end{theorem}

\begin{proof}
The proof that $D^\circ_\cap(A \mid \B) \leq \rho(A, \B)$ is essentially immediate from the proof of Theorem \ref{t:fusion_upper}. It is enough to observe that the construction of $A$ from $\B$ via $\Lambda$ described there can be transformed into a syntactic sequence for $A$ that employs at most $|\Lambda|$ intersection operations. This is similar to the example presented in Section \ref{ss:cyclic}.

We establish next that $\rho(A, \B) \leq D^\circ_\cap(A \mid \B)$. The main difficulty here is that simply unfolding the evaluation of the syntactic sequence introduces further intersection operations (Corollary \ref{c:circ_versus_disc}), and we cannot rely on Theorem \ref{t:fusion_lower}. We argue as follows.

Let $\B = \{B_1, \ldots, B_m\}$, and $I_1, \ldots, I_t$ be a syntactic
sequence that generates $A$ from $\B$ using operations $\star_i$, where $t
= D^\circ(A \mid \B)$. By Lemma \ref{l:convergence}, the evaluation
procedure converges to a sequence $C^1, \ldots, C^m, C^{m+1}, \ldots,
C^{m+t} = A$, where $C^i = B_i$ for $i \in [m]$.
Moreover, each set $I_i$ converges to the set $C^{i+m}$, where $i \in
[t]$.
(This is not an extended
sequence that generates $A$ from $\B$, since the corresponding operations
are not acyclic. However, the relation between the sets is clear.)

\begin{claim}\label{cl:relation_sets}
If $I_i = K_{i_1} \star_i K_{i_2}$ for $i \in [t]$, then $C^{j} = C^{j'}
\diamond_j 
C^{j''}$, where 
$j = i+m$
and
$\diamond_j = \star_i$,
and
$C^{j'}$ and $C^{j''}$
are 
the sets to which $K_{i_1}$ and
$K_{i_2}$ converge, respectively.
\end{claim}

In order to see this, recall that during the evaluation of the syntactic
sequence $I^{\ell+1}_i = I^\ell_{i} \cup (K^\ell_{i_1} \star_i
K^\ell_{i_2})$. Since the evaluation is monotone, and $C^1, \ldots, C^m,
C^{m+1}, \ldots, C^{m+t}$ is the convergent sequence, we eventually have
$I^{\ell+1}_i = I^\ell_{i} = (K^\ell_{i_1} \star_i K^\ell_{i_2})$.
Consequently, $C^{j} = C^{j'} 
\star_i
C^{j''}$ after the indices are
appropriately renamed.\\

For $U = A^c$, let $\Lambda \eqdef \{(C^{j'}_U , C^{j''}_U) \mid j \in \{m+1, \ldots, m+ t\}~\text{and}~\diamond_j = \cap\}$ be a family of pairs of subsets of $U$. In order to complete the proof, it is enough to show that $\Lambda$ covers all semi-filters $\F \subseteq \mathcal{P}(U)$ that are above some element $a = a(\F) \in A$.

Suppose this is not the case, i.e., there is a semi-filter $\F$ above $a \in A$ such that $\F$ is not covered by $\Lambda$. We proceed in part as in the proof of Theorem \ref{t:fusion_lower}. For each $i \in [m + t]$, let $\alpha_i \in \{0,1\}$ be $1$ if and only if $a \in C^{i}$, and $\beta_i \in \{0,1\}$ be $1$ if and only if $C^i_U \in \F$. We obtain a contradiction by a slightly different argument, which is in analogy to the proof in \citep{DBLP:journals/iandc/NakayamaM95}. Since the operations performed over $C^1, \ldots, C^m, C^{m+1}, \ldots, C^{m+t}$ do not follow a linear order, and these sets are obtained after the convergence of the evaluation procedure, we employ a top-down approach, as opposed to the bottom-up presentation that appears in the proof of Theorem \ref{t:fusion_lower}.

We define a partition $(X,Y)$ of the indices of the sets $C^1, \ldots, C^{m + t}$. Note that $\alpha_{m + t} = 1$ and $\beta_{m + t} = 0$ (cf.~Theorem \ref{t:fusion_lower}). Initially, $X$ contains only the element $m + t$. Now for each $j \in X$, if $C^j = C^{j'} \diamond_j C^{j''}$, $\alpha_{j'} = 1$, and $\beta_{j'} = 0$, then we add the element $j'$ to $X$ (and similarly for the index $j''$). We repeat this procedure until no more elements are added to $X$, and let $Y \eqdef [m+t] \setminus X$.

We observe the following properties of this partition.

\begin{claim}\label{cl:base_case}
We have $m + t \in X$ and $\{1, \ldots, m\} \subseteq Y$. If an element $j \in X$, then $\alpha_j = 1$ and $\beta_j = 0$.
\end{claim}

The only non-trivial statement is that $\{1, \ldots, m\} \subseteq Y$. It is enough to argue that if $\ell \in [m]$ then it is not the case that $\alpha_\ell = 1$ and $\beta_\ell = 0$. But since $C^\ell = B_\ell \in \B$ and $\F$ is above $a$, if $\alpha = 1$ (i.e.,~$a \in C^\ell$) then $\beta = 1$ (i.e.,~$B_\ell \cap U \in \F$).

\begin{claim}\label{cl:intersection}
If $j \in X$ and $C^j = C^{j'} \diamond_j C^{j''}$, where $\diamond_j \in \{\cap, \cup\}$ is arbitrary, then either $j' \in X$ or $j'' \in X$.
\end{claim}

Assume contrariwise that $j \in X$ and $j', j'' \in Y$. First, suppose that $\diamond_j = \cap$. Since $\alpha_j = 1$ and $C^j = C^{j'} \cap C^{j''}$, we have $\alpha_{j'} = \alpha_{j''} = 1$. As $j', j'' \in Y$, by construction, we get $\beta_{j'} = \beta_{j''} = 1$ (otherwise one of the indices would be in $X$ and not in $Y$). Consequently, by the definition of the sequence $\beta$, $C^j_U \notin \F$, while $C^{j'}_U, C^{j''}_U \in \F$. This contradictions the assumption that $\Lambda$ does not cover $\F$. Assume now that $\diamond_j = \cup$. Moreover, suppose w.l.o.g.~that $\alpha_{j'} = 1$, which can be done thanks to $C^j = C^{j'} \cup C^{j''}$ and $\alpha_j = 1$. Since $j' \in Y$, we must have $\beta_{j'} = 1$. This means that $C^{j'}_U \in \F$, and by the monotonicity of $\F$ and $\diamond_j = \cup$, it follows that $C^{j}_U \in \F$. But this is in contradiction to $\beta_j = 0$, which completes the proof of the claim.

\begin{claim}\label{cl:union}
Suppose that $j, j' \in X$, $C^j = C^{j'} \cup C^{j''}$, and $j'' \in Y$. Then $a \notin C^{j''}$.
\end{claim}

The assumptions force $\alpha_j = 1$ and $\beta_j = 0$, and that it is not the case that $\alpha_{j''} = 1$ and $\beta_{j''} = 0$. We must argue that $\alpha_{j''} = 0$ (i.e.,~$a \notin C^{j''}$), and to do so we show that $\beta_{j''} = 0$. But if $\beta_{j''} = 1$, the monotonicity of $\F$ and $C^j = C^{j'} \cup C^{j''}$ imply $\beta_j = 1$, a contradiction. This completes the proof of this claim.\\

Finally, we combine these three claims, derived from the assumption that there is a semi-filter $\F$ above $a$ that is not covered by $\Lambda$, to get a contradiction. Recall that $C^1, \ldots, C^{m + t} = A$ is the convergent sequence obtained from the syntactic sequence $I_1, \ldots, I_t$ and its operations $\star_i$, and that by assumption $a \in A$. Therefore, our proof will be complete if we can show that $a \notin C^{m + t}$. 

In order to establish this final implication, we show the stronger statement that the element $a$ is never added to a set $C^j$ during the update steps of the evaluation procedure if $j \in X$ (since $m + t \in X$ by Claim \ref{cl:base_case}), which is a contradiction. Before the first update, each such set is empty, as the only non-empty sets are in $\B$, and these have indices in $Y$ (Claim \ref{cl:base_case}). During an update of the elements of a set $C^j$ with $j \in X$, we consider two cases based on $\diamond_j \in \{\cup , \cap\}$. If $\diamond_j = \cap$, Claim \ref{cl:intersection} implies that at least one of the operands comes from $X$, and thus by induction the update step will not include $a$ in $C^j$. On the other hand, if $\diamond_j = \cup$, Claim \ref{cl:intersection} shows that at most one operand comes from $Y$. If there is no operand from $Y$, we are done using the induction hypothesis. Otherwise, Claim \ref{cl:union} implies that $a$ is not an element of this operand (as it is not in the corresponding set even after the evaluation procedure converges). By the induction hypothesis, $a$ is not added to $C^j$. This finishes the proof of Theorem \ref{t:fusion_cyclic}.    
\end{proof}

In particular, this result shows that the $k$-clique lower bound discussed
in \citep{DBLP:conf/coco/Karchmer93} holds in the more general model of
cyclic  Boolean circuits.
Indeed,
Karchmer shows a lower bound for 
$\rho(A,\calb)$,
where $A$ is the set of graphs with $k$-cliques and $\calb$ is the
monotone Boolean basis. Combined with the previous result, this gives the following
corollary.

\begin{corollary}[Consequence of Theorem \ref{t:fusion_cyclic} and \citep{DBLP:conf/coco/Karchmer93}]
Let $k$-$\mathsf{clique}\colon \{0,1\}^{\binom{n}{2}} \to \{0,1\}$ be the function that evaluates to $1$ on an undirected $n$-vertex input graph $G$ if and only if $G$ contains a $k$-clique. Then every monotone cyclic Boolean circuit that computes $3$-$\mathsf{clique}$ contains at least $\Omega(n^3/(\log n)^4)$ fan-in two \emph{AND} gates.
\end{corollary}

This lower bound against monotone cyclic circuits does not seem to easily follow from the proofs in \citep{razborov1985lower, DBLP:journals/combinatorica/AlonB87}. 

\section{Graph Complexity and Two-Dimensional Cover Problems}\label{s:graph_complex}

\subsection{Basic results and connections}\label{ss:fram_cover_complex}

\begin{proposition}[The intersection complexity of a random graph]\label{p:intersection_random_graph}
Let $G \subseteq [N] \times [N]$ be a random bipartite graph. Then, asymptotically almost surely,
$$
D_\cap(G \mid \G_{N,N}) \;=\;\Theta(N).
$$
\end{proposition}

\begin{proof}
The upper bound is easy, and holds in the worst case as well (see Section \ref{ss:graph_complexity}). For the lower bound, recall that a random graph $G$ satisfies $D(G \mid \G_{N,N}) = \Omega(N^2/\log N)$, which is an immediate consequence of Lemma \ref{l:complex_sets}. By Lemma \ref{l:zwick}, it must be the case that $D_\cap(G \mid \G_{N,N}) \;=\;\Omega(N)$, which completes the proof.
\end{proof}

Recall the definition of cover complexity introduced in Section \ref{ss:fusion_defin_notat}. Theorem  \ref{t:fusion_upper} and Proposition \ref{p:intersection_random_graph} yield an $\Omega(\sqrt{N})$ lower bound on the cover complexity of a random graph. It is possible to obtain a tight lower bound using a more careful argument.

\begin{theorem}[The cover complexity of a random graph]\label{t:cover_complex_random_graph} Let 
    $G \subseteq [N] \times [N]$ be a random bipartite graph. Then, asymptotically almost surely,
$$
\rho(G, \G_{N,N}) \; = \; \Theta(N).
$$ 
\end{theorem}

\begin{proof}
The proof is based on a counting argument, and can be formalized using Kolmogorov complexity. Observe that the proof of Theorem \ref{t:fusion_upper} describes a \emph{universal procedure} that generates an \emph{arbitrary} set $A$ from $\B$ using $\Lambda$. However, for a \emph{fixed} family $\B$ such as  $\B = \G_{N,N}$, the only information the procedure needs is the inclusion relation among the sets appearing in $\Lambda$ and $\B$. Crucially, the explicit description of the sets that appear in $\Lambda$ is not necessary to fully specify the corresponding set $A$ that is generated by the universal procedure. Indeed, observe that the core of the construction after the base case (which does not depend on $A$) are the sub-indices appearing in Equations \ref{eq:num1}, \ref{eq:num2}, and \ref{eq:num3}, which are determined by the aforementioned inclusion relations.  These inclusions can be described by $O(|\Lambda|(|\B| + |\Lambda|))$ bits. Since a random graph has description complexity $\Omega(N^2)$ and $|\G_{N,N}| = 2N$, we must have $|\Lambda| = \Omega(N)$ asymptotically almost surely. In other words, $\rho(G, \G_{N,N}) = \Omega(N)$ for a typical graph $G \subseteq [N] \times [N]$. 
\end{proof}

Let $N = 2^n$. For a graph $G \subseteq [N] \times [N]$, we let $f_G \colon \{0,1\}^{2n} \to \{0,1\}$ be the Boolean function associated with $G$, as described in Lemma \ref{l:graph_to_circuit} (in other words, $f_G^{-1}(1) = \phi(G)$).

\begin{proposition}[Reducing circuit complexity lower bounds to two-dimensional cover problems]\label{p:graph_approach}
For any non-trivial graph $G \subseteq [N] \times [N]$,
$$
\rho(G, \G_{N,N}) \;\leq\;  D_\cap(f_G^{-1}(1) \mid \B_{2n}).
$$
\end{proposition}

\begin{proof}
This follows from Theorem \ref{t:fusion_lower} and Lemma \ref{l:graph_to_circuit}.
\end{proof}

These results do not immediately imply that $\rho(G, \G_{N,N}) \leq \rho(f_G^{-1}(1) , \B_{2n})$, since the connection between $D_\cap$ and $\rho$ might not be tight. This can be shown by a direct argument.

\begin{lemma}[A fusion transference  lemma]\label{l:fusion_magnification}
Let $G \subseteq [N] \times [N]$ be a non-trivial graph. Then,
$$
\rho(G, \G_{N,N}) \;\leq\; \rho(f_G^{-1}(1) , \B_{2n}).
$$
\end{lemma}

\begin{proof}
Let $\mathfrak{F}^\uparrow_{f_G}$ be the set that contains a semi-filter $\F$ over $f^{-1}_G(0)$ if and only if it is above some element $a \in f^{-1}_G(1)$. Similarly, let $\mathfrak{F}^\uparrow_G$ contain a semi-filter $\F$ over $\overline{G}$ if and only if there is $(u,v) \in G$ such that $\F$ is above $(u,v)$. Assume $\Lambda_{f_G}$ is a family of pairs of subsets of $f^{-1}_G(0)$ that cover all semi-filters in $\mathfrak{F}^\uparrow_{f_G}$. Now let $\Lambda_G$ be the family of pairs of subsets of $\overline{G}$ induced by the pairs in $\Lambda_{f_G}$ and the bijection between $[N] \times [N]$ and $\{0,1\}^{2n}$. We claim that $\Lambda_G$ covers all semi-filters in $\mathfrak{F}^\uparrow_G$.\footnote{Note that the semi-filters in $\mathfrak{F}_{f_G}^\uparrow$ and in $\mathfrak{F}_G^\uparrow$ differ in their definitions of ``above'', as they are connected to different sets of generators.} 

Recall that we identify an element $(u,v) \in [N] \times [N]$ with its corresponding input string $\phi(u,v) = \binary(u)\binary(v) \in \{0,1\}^{2n}$, which for convenience we will simply denote by $uv$. Assume this is not the case, i.e., there is a semi-filter $\F \in \mathfrak{F}^\uparrow_G$ that is above some edge $(u,v) \in G$ and preserves $\Lambda_G$ (in other words, it is not covered by $\Lambda_G$).   Let $\F'$ be the corresponding family of subsets of $f^{-1}_G(0)$ under $\phi$. Observe that $\F'$ is a semi-filter over $f^{-1}_G(0)$, and that it preserves $\Lambda_{f_G}$. Therefore, in order to get a contradiction it is enough to verify that $\F'$ is above $uv$ (with respect to the family of generators $\mathfrak{B}_{2n} \subseteq \mathcal{P}(\{0,1\}^{2n})$). This follows easily using the upward-closure of $\F$ and the fact that $\F$ is above the edge $(u,v)$ with respect to $\G_{N,N}$, as we explain next. 

For instance, assume that $u_i = 0$ for some $i \in [n]$. We must prove that the corresponding set $B_i^c \cap f^{-1}_G(0) \in \F'$. 
From $u_i = 0$, we get $R_u \subseteq \phi^{-1}(B_i^c)$, and then $R_u \cap \overline{G} \subseteq \phi^{-1}(B_i^c) \cap \overline{G} = \phi^{-1}(B_i^c \cap f_G^{-1}(0))$. Since $\F$ is above $(u,v)$ with respect to $\G_{N,N}$, $R_u \cap \overline{G} \in \F$. Consequently, $\phi(R_u \cap \overline{G}) \in \F'$. Now $\phi(R_u \cap \overline{G}) \subseteq \phi(\phi^{-1}(B_i^c \cap f_G^{-1}(0))) = B_i^c \cap f^{-1}_G(0)$, and from the upward-closure of $\F'$, the latter set is in $\F'$. The remaining cases are similar.
\end{proof}

This result and Theorem \ref{t:fusion_lower} provide an alternative proof of Proposition \ref{p:graph_approach}. As we will see later in this section, establishing a direct connection among cover problems can have further benefits (Section \ref{ss:nondet_graph_complexity}).

\subsection{A simple lower bound example}\label{ss:example_fusion_lb}

Let $N = 2^n$. Consider the graph $G_\mathsf{NEQ} \subseteq [N] \times [N]$, where $(u,v) \in G_\mathsf{NEQ}$ if and only if $u \neq v$. \Cref{fig:NEQ_picture} below describes the $N = 8$ case. We show a tight lower bound on $\rho(G_\mathsf{NEQ}, \G_{N,N})$. To prove this result, we focus on a particular set of semi-filters. For convenience, we write $G = G_\mathsf{NEQ}$. 

\begin{figure}[htbp]
  \centering
\begin{tikzpicture}[scale=0.5]
    \def\N{8}
    
    \foreach \i in {1,...,\N} {
        \foreach \j in {1,...,\N} {
            \ifnum\i=\j
                \fill[white] (\j-1,-\i+1) rectangle (\j,-\i);
            \else
                \fill[lightgray] (\j-1,-\i+1) rectangle (\j,-\i);
            \fi
        }
    }
    
    \draw[thick] (0,0) grid (\N,-\N);

    \draw[very thick] (0,0) rectangle (\N,-\N);
\end{tikzpicture}
\caption{A graphical representation of $G_\mathsf{NEQ} \subseteq [N] \times [N]$ for $N = 8$. \Cref{prop:NEQ_lb} shows that for this value of $N$ the intersection complexity is $3$.}
  \label{fig:NEQ_picture}
\end{figure}
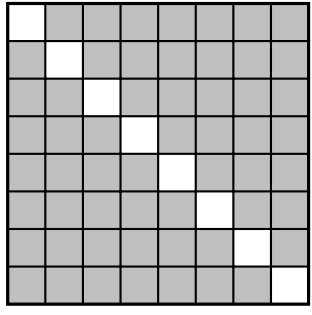

For $e \in G$, where $e = (u,v)$, we let $\F_e$ be the upward closure (with
respect to $\overline{G})$ of the family that contains the sets
$R^u_{\overline{G}}$ and $C^v_{\overline{G}}$,
where
$R^u_{\overline{G}} = R_u \cap \overline{G}$
and
$C^v_{\overline{G}} = C_v \cap \overline{G}$.
More explicitly, a set $W$ is in $\F_e$
iff
$R^u_{\overline{G}} \sseq W$
\emph{or}
$C^v_{\overline{G}} \sseq W$.
Notice that, in general (i.e.,~for an arbitrary graph), this might not be a semi-filter, as one of the sets might be empty. But for our choice of $G$, this is a semi-filter above $e$. We let
$$
\mathfrak{F}_{\mathsf{can}}^G \eqdef \{\F_e \mid e \in G~\text{and}~\F_e~\text{is a semi-filter}\;\}.
$$
We say that $\mathfrak{F}_{\mathsf{can}}^G$ is the set of \emph{canonical semi-filters} of $G$ (above an edge of $G$). 
Note that $\mathfrak{F}_{\mathsf{can}}^{G^*}$ 
is well-defined
for any bipartite graph $G^* \sseq [N] \times [N]$.
We can thus ask, for a general bipartite graph $G^*$,
how many
pairs of subsets of $\overline{G^*}$ are needed to cover all semi-filters in
$\mathfrak{F}_{\mathsf{can}}^{G^*}$? Let us denote this quantity by
$\rho_\mathsf{can}(G^*,\G_{N,N})$, i.e., the \emph{canonical cover
complexity} of $G^*$. Clearly, this quantity lower bounds cover complexity.

\begin{proposition} \label{prop:NEQ_lb} For the graph $G = G_\mathsf{NEQ}$ defined above,
$$\rho_\mathsf{can}(G, \G_{N,N}) = \rho(G, \G_{N,N}) = D_\cap(G \mid \G_{N,N}) = n = \log N.$$
\end{proposition}

\begin{proof}
The upper bound follows by transforming a circuit for the corresponding Boolean function $f_G \colon \{0,1\}^{n} \times \{0,1\}^n \to \{0,1\}$ into a construction of $G$. Observe that $f_G(u,v) = \bigvee_{i \in [n]} u_i \oplus v_i$, where $\oplus$ denotes the parity operation, and that each $\oplus$-gate can be implemented using a single $\wedge$-gate via $a \oplus b = (a \vee b) \wedge (\overline{a} \vee \overline{b})$. Therefore, $\rho_\mathsf{can}(G, \G_{N,N}) \leq \rho(G, \G_{N,N}) \leq D_\cap(G \mid \G_{N,N}) \leq n$ via Lemma \ref{l:graph_to_circuit} and Theorem \ref{t:fusion_lower}. 

For the lower bound on $\rho_\mathsf{can}(G, \G_{N,N})$, let $\Lambda = \{(E_1, H_1), \ldots, (E_k,H_k)\}$ be a family of $k$ pairs of subsets of $\overline{G}$. We argue that if $\Lambda$ covers all semi-filters in $\mathfrak{F}_\mathsf{can}^G$ then $k \geq n$. Recall that, for every $e \in G$, $\mathcal{F}_e$ is a semi-filter above $e$, i.e., $\F_e \in \mathfrak{F}_\mathsf{can}^G$. Fix a pair $(E,H) \in \Lambda$.

\begin{claim}\label{c:egde_elim}
Let $e = (u,v) \in G$, and $\F_{e} \in \mathfrak{F}_\mathsf{can}^G$. Then $\F_e$ is covered by $(E,H)$ if and only if each singleton set $R^u_{\overline{G}}$ and $C^v_{\overline{G}}$ is contained in precisely one of $E$ and $H$, and none of the latter sets contains both of them.
\end{claim}

    First, we argue that $\F_e$ is covered under the condition in the claim. Assume without loss of generality that $R^u_{\overline{G}} \subseteq E$ and $C^v_{\overline{G}} \subseteq H$. Then, using the definition of $\F_e$, we get that $E \in \F_e$ and $H \in \F_e$. On the other hand, by assumption, $R^u_{\overline{G}} \nsubseteq E \cap H$ and $C^v_{\overline{G}} \nsubseteq E \cap H$. This implies that $E \cap H \notin \F_e$. In other words, $(E,H)$ covers $\F_e$.

Suppose now that $(E,H)$ covers $\F_e$. Then $E, H \in \F_e$ but $E \cap H
\notin \F$. It is easy to check that this implies the condition in the
statement of Claim \ref{c:egde_elim}.\\

Claim \ref{c:egde_elim} immediately implies the following lemma.

\begin{lemma}
Every semi-filter in $\mathfrak{F}_\mathsf{can}^G$ covered by $(E,H)$ is also covered by $(E \setminus H, H \setminus E)$.
\end{lemma}

Thus we can and will assume w.l.o.g.~that all pairs appearing in $\Lambda$ have disjoint sets $E_i$ and $H_i$. Using Claim \ref{c:egde_elim} again, we  obtain the following additional consequence.

\begin{lemma}
Every semi-filter in $\mathfrak{F}_\mathsf{can}^G$ covered by a disjoint pair $(E,H)$ is also covered by the pair $(E, \overline{G} \setminus E)$.
\end{lemma}

Consequently, we will further assume that all pairs appearing in $\Lambda$ form a partition of $\overline{G}$. Let $(E_1, H_1) \in \Lambda$ be one such pair. Since $E_1$ and $H_1$ form a partition of $\overline{G}$, either $|E_1| \geq N/2$ or $|H_1| \geq N/2$. Assume w.l.o.g~that $|E_1| \geq N/2$. Let $G_1 \subseteq G$ be the subgraph of $G$ obtained when the ambient space $[N] \times [N]$ is restricted to $\mathsf{Rows}(E_1) \times \mathsf{Columns}(E_1)$, where $\mathsf{Rows}(E_1) = \{a \in [N] \mid (a,b) \in E_1~\text{for some}~b \in [N]\}$, and $\mathsf{Columns}(E_1)$ is defined analogously.

 Observe that for no element $e_1 \in G_1$, $\F_{e_1}$ is covered by $(E_1, H_1)$. Furthermore, the elements in $G_1$ span at least $2^{n -1}$ different rows and at least $2^{n -1}$ different columns of $[N]$. Finally, each semi-filter $\F_{e_1} \in \mathfrak{F}^G_\mathsf{can}$ for $e_1 \in G_1$ must be covered by some pair in $\Lambda \setminus \{(E_1, H_1)\}$. By a recursive application of the previous argument, and using that in the base case $n = 1$ at least one pair of sets is necessary, it is easy to see $|\Lambda| \geq n = \log N$. This completes the proof.
\end{proof}

\subsection{Nondeterministic graph complexity}\label{ss:nondet_graph_complexity}

Given a Boolean function $f \colon \{0,1\}^n \to \{0,1\}$,  we let $\mathsf{size}(f)$ be the minimum number of fan-in two AND/OR gates in a DeMorgan Boolean circuit computing $f$ (we assume negations appear only at the input level). We can define $\mathsf{size}_\vee(f)$ and $\mathsf{size}_\wedge(f)$ in a similar way. Using our notation, $\mathsf{size}(f) = D(f\mid \mathcal{B}_n)$, $\mathsf{size}_\vee(f) = D_\cup(f\mid \mathcal{B}_n)$, and $\mathsf{size}_\wedge(f) = D_\cap(f\mid \mathcal{B}_n)$.

We also define $\mathsf{conondet}$-$\mathsf{size}_\wedge(f)$ to be the minimum number of $\wedge$-gates in a circuit $D(x,y)$ such that $f(x) = 1$ if and only if for all $y$ we have $D(x,y) = 1$. Similarly, $\mathsf{nondet}$-$\mathsf{size}_\vee(g)$ is the minimum number of $\vee$-gates in a circuit $C(x,y)$ such that $g(x) = 1$ if and only if there exists $y$ such that $C(x,y) = 1$. Observe that for every Boolean function $h$, $\mathsf{conondet}$-$\mathsf{size}_\wedge(h) = \mathsf{nondet}$-$\mathsf{size}_\vee(\neg h)$.

Observe that the definition of nondeterministic complexity for Boolean functions relies on Boolean circuits extended with extra input variables. It is not entirely clear how to introduce a natural similar definition in the context of graph complexity, i.e,~a nondeterministic version of $D(G \mid \G_{N,N})$. We take a different path, and translate an alternative characterization of nondeterministic complexity in the Boolean function setting (based on the fusion method) to the graph complexity setting. First, we review the necessary concepts.

\begin{definition}[Semi-ultra-filter]
We say that a semi-filter $\F \subseteq \mathcal{P}(U)$ is a \emph{semi-ultra-filter} if for every set $A \subseteq U$, at least one of $A$ or $U \setminus A$ is in $\F$.
\end{definition}

For a function $f \colon \{0,1\}^n \to \{0,1\}$, let $\rho_\mathsf{ultra}(f, \B_n)$ denote the minimum number of pairs of subsets of $f^{-1}(0)$ that cover all semi-ultra-filters over $f^{-1}(0)$ that are above an input in $f^{-1}(1)$. \citep{DBLP:conf/coco/Karchmer93} established the following result.

\begin{theorem}
There exists a constant $c \geq 1$ such that for every function $f \colon \{0,1\}^n \to \{0,1\}$, 
$$
\rho_\mathsf{ultra}(f, \B_n) \;\leq\; \mathsf{conondet}\text{-}\mathsf{size}_\wedge(f) \;=\; \mathsf{nondet}\text{-}\mathsf{size}_\vee(\neg f)  \;\leq\; c \cdot \rho_\mathsf{ultra}(f, \B_n).
$$
\end{theorem}

Roughly speaking, a variation of cover complexity can be used to characterize conondeterministic circuit complexity. This motivates the following definition, which provides a notion of nondeterministic complexity in arbitrary discrete spaces.

\begin{definition}[Conondeterministic cover complexity]
Given a discrete space $\langle \amb, \mathcal{B} \rangle$ and a set $A \subseteq \amb$, we let $\rho_\mathsf{ultra}(A, \B)$ denote the minimum number of pairs of subsets of $U = A^c = \amb \setminus A$ that cover all semi-ultra-filters over $U$ that are above an element $a \in A$.
\end{definition}

Observe that $\rho_\mathsf{ultra}(A, \B) \leq \rho(A, \B)$, since every semi-ultra-filter is a semi-filter. Conondeterministic cover complexity sheds  light into the strength of the simple lower bound argument presented in Section \ref{ss:example_fusion_lb}.

\begin{proposition}
Let $G_{\mathsf{NEQ}} \subseteq [N] \times [N]$ be the graph defined in Section \ref{ss:example_fusion_lb}. Then,
$$
\rho_\mathsf{can}(G_{\mathsf{NEQ}}, \G_{N,N}) \; \leq \; \rho_\mathsf{ultra}(G_{\mathsf{NEQ}}, \G_{N,N}).
$$
\end{proposition}

\begin{proof}
For convenience, let $G = G_{\mathsf{NEQ}}$. Simply observe that every semi-filter $\mathcal{F}_e$ in $\mathfrak{F}^G_\mathsf{can}$ is a semi-ultra-filter. Indeed, for $e = (u,v) \in G$ and an arbitrary set $W \subseteq \overline{G}$, either $W$ or $\overline{G} \setminus W$ contains $R^u_{\overline{G}}$, since the latter is a singleton set due to our choice of $G$.
\end{proof}

Now we translate this result into a stronger lower bound in Boolean function complexity. This will be a consequence of the following lemma.

\begin{lemma}[A nondeterministic fusion transference lemma] \label{l:nondet_fusion_magnification}~\\
Let $N = 2^n$. For every graph $G \subseteq [N] \times [N]$, 
$$
\rho_\mathsf{ultra}(G, \G_{N,N}) \;\leq \; \rho_\mathsf{ultra}(f_G, \B_{2n}),
$$ 
where $f \colon \{0,1\}^{2n} \to \{0,1\}$ is the Boolean function associated with $G$.
\end{lemma}

\begin{proof}
Recall that, in the proof of Lemma \ref{l:fusion_magnification} (fusion transference lemma), if a semi-filter $\F$ in the graph setting is not covered, then it gives rise to a semi-filter $\F'$ in the Boolean function setting that is not covered. Crucially, if the original semi-filter is a semi-ultra-filter, so is the resulting semi-filter. The proof of this fact is obvious, since $\phi \colon [N] \times [N] \to \{0,1\}^{2n}$ is a bijection.
\end{proof}

Let $\mathsf{NEQ}_{2n} \colon \{0,1\}^n \times \{0,1\}^n \to \{0,1\}$ be the function such that $\mathsf{NEQ}_{2n}(x,y) = 1$ if and only if $x \neq y$, and $\mathsf{EQ}_{2n}$ be its negation. By combining the ideas of this section and Section \ref{ss:example_fusion_lb}, we get the following tight inequalities.

\begin{corollary}[A simple nondeterministic lower bound via graph complexity + fusion]\label{c:tight_relations}
\begin{eqnarray}
n & \leq & \rho_\mathsf{can}(G_\mathsf{NEQ}, \G_{N,N}) \nonumber \\ 
  & \leq & \rho_\mathsf{ultra}(G_\mathsf{NEQ}, \G_{N,N}) \nonumber \\
  & \leq & \rho_\mathsf{ultra}(\mathsf{NEQ}_{2n}, \B_{2n}) \nonumber \\ 
  & \leq & \mathsf{conondet}\text{-}\mathsf{size}_\wedge(\mathsf{NEQ}_{2n}) \nonumber \\
  & \leq & \mathsf{nondet}\text{-}\mathsf{size}_\vee(\mathsf{EQ}_{2n}) \nonumber \\ 
  & \leq & \mathsf{size}_\vee(\mathsf{EQ}_{2n}) \nonumber \\ 
  & \leq & \mathsf{size}_\wedge(\mathsf{NEQ}_{2n}) \nonumber \\
  & \leq & n. \nonumber 
\end{eqnarray}
\noindent In particular, the nondeterministic union complexity of the Boolean function $\mathsf{EQ}_{2n}$ is precisely $n$.
\end{corollary}

Observe that, by Theorem \ref{t:fusion_cyclic}, a cyclic circuit computing $\mathsf{NEQ}_{2n}$ also requires $n$ fan-in two AND gates.

\bibliographystyle{alpha}
\bibliography{refs}

\begin{thebibliography}{GHKK16}

\bibitem[AB87]{DBLP:journals/combinatorica/AlonB87}
Noga Alon and Ravi~B. Boppana.
\newblock The monotone circuit complexity of {B}oolean functions.
\newblock {\em Combinatorica}, 7(1):1--22, 1987.

\bibitem[Cha94]{CHASHKIN+1994+229+258}
A.~V. Chashkin.
\newblock On the complexity of boolean matrices, graphs, and the boolean
  functions corresponding to them.
\newblock {\em Discrete Mathematics and Applications}, 4(3):229--258, 1994.

\bibitem[FGHK16]{DBLP:conf/focs/FindGHK16}
Magnus~Gausdal Find, Alexander Golovnev, Edward~A. Hirsch, and Alexander~S.
  Kulikov.
\newblock A better-than-3n lower bound for the circuit complexity of an
  explicit function.
\newblock In {\em Symposium on Foundations of Computer Science \emph{(FOCS)}},
  pages 89--98, 2016.

\bibitem[GH08]{givant2008introduction}
Steven Givant and Paul Halmos.
\newblock {\em Introduction to {B}oolean algebras}.
\newblock Springer, 2008.

\bibitem[GHKK16]{DBLP:conf/mfcs/GolovnevHKK16}
Alexander Golovnev, Edward~A. Hirsch, Alexander Knop, and Alexander~S. Kulikov.
\newblock On the limits of gate elimination.
\newblock In {\em International Symposium on Mathematical Foundations of
  Computer Science \emph{(MFCS)}}, pages 46:1--46:13, 2016.

\bibitem[Gol18]{GolovnevPrivate18}
Alexander Golovnev.
\newblock Private communication, 2018.

\bibitem[Juk12]{DBLP:books/daglib/0028687}
Stasys Jukna.
\newblock {\em {B}oolean Function Complexity - Advances and Frontiers}.
\newblock Springer, 2012.

\bibitem[Juk13]{juknacomputational}
Stasys Jukna.
\newblock Computational complexity of graphs (book chapter).
\newblock {\em Advances in Network Complexity, Quantitative and Network
  Biology}, pages 99--153, 2013.

\bibitem[Kar93]{DBLP:conf/coco/Karchmer93}
Mauricio Karchmer.
\newblock On proving lower bounds for circuit size.
\newblock In {\em Structure in Complexity Theory Conference \emph{(CCC)}},
  pages 112--118, 1993.

\bibitem[Lok03]{DBLP:journals/mst/Lokam03}
Satyanarayana~V. Lokam.
\newblock Graph complexity and slice functions.
\newblock {\em Theory Comput. Syst.}, 36(1):71--88, 2003.

\bibitem[LY22]{DBLP:conf/stoc/Li022}
Jiatu Li and Tianqi Yang.
\newblock 3.1\emph{n} - \emph{o}(\emph{n}) circuit lower bounds for explicit
  functions.
\newblock In {\em Symposium on Theory of Computing \emph{(STOC)}}, pages
  1180--1193, 2022.

\bibitem[NM95]{DBLP:journals/iandc/NakayamaM95}
Katsutoshi Nakayama and Akira Maruoka.
\newblock Loop circuits and their relation to {R}azborov's approximation model.
\newblock {\em Inf. Comput.}, 119(2):154--159, 1995.

\bibitem[Oli18]{notes_approx}
Igor~C. Oliveira.
\newblock Notes on the method of approximations and the emergence of the fusion
  method.
\newblock {\em \emph{Manuscript (available online)}}, 2018.

\bibitem[PRS88]{DBLP:journals/acta/PudlakRS88}
Pavel Pudl{\'{a}}k, Vojtech R{\"{o}}dl, and Petr Savick{\'{y}}.
\newblock Graph complexity.
\newblock {\em Acta Inf.}, 25(5):515--535, 1988.

\bibitem[Raz85]{razborov1985lower}
Alexander~A. Razborov.
\newblock Lower bounds for the monotone complexity of some {B}oolean functions.
\newblock {\em Soviet Math. Doklady}, 31:354--357, 1985.

\bibitem[Raz89]{DBLP:conf/stoc/Razborov89}
Alexander~A. Razborov.
\newblock On the method of approximations.
\newblock In {\em Symposium on Theory of Computing \emph{(STOC)}}, pages
  167--176, 1989.

\bibitem[RM99]{DBLP:journals/combinatorica/RazM99}
Ran Raz and Pierre McKenzie.
\newblock Separation of the monotone {NC} hierarchy.
\newblock {\em Combinatorica}, 19(3):403--435, 1999.

\bibitem[Sch88]{DBLP:conf/aaecc/Schnorr88}
Claus{-}Peter Schnorr.
\newblock The multiplicative complexity of {B}oolean functions.
\newblock In {\em International Conference on Applied Algebra, Algebraic
  Algorithms and Error-Correcting Codes \emph{(AAECC)}}, pages 45--58, 1988.

\bibitem[Wig93]{Wigderson93thefusion}
Avi Wigderson.
\newblock The fusion method for lower bounds in circuit complexity.
\newblock In {\em Combinatorics, Paul Erdos is Eighty, Bolyai Math. Society},
  pages 453--467, 1993.

\bibitem[Wig95]{wigderson1995lectures}
Avi Wigderson.
\newblock Lectures on the fusion method and derandomization.
\newblock {\em Technical Report}, 1995.

\bibitem[Zwi96]{DBLP:journals/ipl/Zwick96}
Uri Zwick.
\newblock On the number of {AND}s versus the number of {OR}s in monotone
  {B}oolean circuits.
\newblock {\em Inf. Process. Lett.}, 59(1):29--30, 1996.

\end{thebibliography}

\end{document}